\documentclass[10pt,aps,prl,twocolumn,longbibliography,floatfix,nobalancelastpage,superscriptaddress]{revtex4-1}

\usepackage[T1]{fontenc}
\usepackage{graphicx}
\usepackage{amssymb}
\usepackage{mathtools}
\usepackage{amsmath}
\usepackage{amsthm}
\usepackage{xr-hyper}
\usepackage[colorlinks=true,linkcolor=blue,citecolor=blue,urlcolor=blue,plainpages=false,pdfpagelabels]{hyperref}
\usepackage{color,xcolor,colortbl}
\usepackage{multirow}
\usepackage{bbm}
\usepackage{float}
\usepackage{array}
\usepackage{verbatim}
\usepackage{braket}
\usepackage{dsfont}

\usepackage{times}
\usepackage{txfonts}
\usepackage[mathcal]{euscript}

\DeclareMathOperator{\Tr}{Tr}

\newcommand{\I}{\mathrm{i}}
\newcommand{\e}{\mathrm{e}}
\newcommand{\G}{\mathrm{\scriptscriptstyle G}}
\renewcommand{\t}{{\scriptscriptstyle\mathsf{T}}}

\newcommand{\ketbra}[1]{\ket{#1}\!\!\bra{#1}}
\newcommand{\ketbraa}[2]{\ket{#1}\!\!\bra{#2}}

\newcommand{\texteq}[1]{\overset{#1}{=}}
\newcommand{\textleq}[1]{\overset{#1}{\leq}}

\newcommand{\textgeq}[1]{\overset{#1}{\geq}}

\newcommand{\lsmatrix}{\left(\begin{smallmatrix}}
\newcommand{\rsmatrix}{\end{smallmatrix}\right)}

\makeatletter
\g@addto@macro\bfseries{\boldmath}
\makeatother

\theoremstyle{plain}%
\newtheorem{theorem}{Theorem}
\newtheorem*{theorem*}{Theorem}
\newtheorem{lemma}{Lemma}
\newtheorem{corollary}[theorem]{Corollary}

\newtheorem{proposition}[theorem]{Proposition}
\newtheorem{remark}[theorem]{Remark}

\theoremstyle{definition}

\allowdisplaybreaks

\begin{document}

\widetext

\title{Extendibility of bosonic Gaussian states}

\author{Ludovico Lami}
\email{ludovico.lami@gmail.com}
\affiliation{School of Mathematical Sciences and Centre for the Mathematics and Theoretical Physics of Quantum Non-Equilibrium Systems, University of Nottingham, University Park, Nottingham NG7 2RD, United Kingdom}

\author{Sumeet Khatri}
\email{skhatr5@lsu.edu}
\affiliation{Hearne Institute for Theoretical Physics, Department of Physics and Astronomy, Louisiana State University, Baton Rouge, Louisiana, 70803, USA}

\author{Gerardo Adesso}
\email{gerardo.adesso@nottingham.ac.uk}
\affiliation{School of Mathematical Sciences and Centre for the Mathematics and Theoretical Physics of Quantum Non-Equilibrium Systems, University of Nottingham, University Park, Nottingham NG7 2RD, United Kingdom}

\author{Mark M. Wilde}
\email{mwilde@lsu.edu}
\affiliation{Hearne Institute for Theoretical Physics, Department of Physics and Astronomy, Louisiana State University, Baton Rouge, Louisiana, 70803, USA}
\affiliation{Center for Computation and Technology, Louisiana State University, Baton Rouge, Louisiana, 70803, USA}

\date{\today}

\begin{abstract}
	
Extendibility of bosonic Gaussian states is a key issue in continuous-variable quantum information. We show that a bosonic Gaussian state is $k$-extendible if and only if it has a Gaussian $k$-extension, and we derive a simple semidefinite program, whose size scales linearly with the number of local modes, to efficiently decide $k$-extendibility of any given bosonic Gaussian state. When the system to be extended comprises one mode only, we provide a closed-form solution. Implications of these results  for the steerability of quantum states and for the extendibility of bosonic Gaussian channels are discussed. We then derive upper bounds on the distance of a $k$-extendible bosonic Gaussian state to the set of all separable states, in terms of trace norm and R\'enyi relative entropies. These bounds, which can be seen as ``Gaussian de Finetti theorems,'' exhibit a universal scaling in the total number of modes, independently of the mean energy of the state. Finally, we establish an upper bound on the entanglement of formation of Gaussian $k$-extendible states, which has no analogue in the finite-dimensional setting.

\end{abstract}

\maketitle

Entanglement is the mainspring of modern quantum technologies. To tally the performance of such technologies, a comprehensive characterization and quantification of entanglement is needed. One of the defining features of entanglement is its {\em monogamy}~\cite{terhal2004,coffman2000,FLV88, RW89, DPS02, complete-extendibility,lancien2016}, the fact that entangled states cannot be shared among arbitrarily many subsystems. Exploring the middle ground of partially shareable states or, precisely, partially {\em extendible} states, offers a rich and practically meaningful lookout into the virtues of entanglement as a resource.

A bipartite quantum state $\rho_{AB}$ of systems $A$ and $B$ is called \textit{$k$-extendible} on $B$ if there exists a quantum state $\widetilde{\rho}_{AB_1\dotsb B_k}$ on  $A$ and $k$ copies $B_1,\dotsc,B_k$ of $B$ that is permutation-invariant with respect to the systems $B_i$ and satisfies $\Tr_{B_2\dotsb B_k}\left[\,\widetilde{\rho}_{AB_1\dotsb B_k}\right]=\rho_{AB}$, where $B_1\equiv B$. It is well-known that a state $\rho_{AB}$ is separable if and only if it is $k$-extendible for all $k\geq 2$~\cite{FLV88, RW89, DPS02, complete-extendibility}. 
The nested sets of $k$-extendible states can thus be used to approximate the set of separable states, which has resulted in work on quantum de Finetti theorems~\cite{deFinetti0, deFinetti1, 1-1/2-de-Finetti, deFinetti2, Koenig2009, deFinetti4, CJYZ16} and other studies of entanglement~\cite{NOP09, BC12}. Extendibility also arises in the contexts of security of quantum key distribution~\cite{MCL06, MRDL09, KL17}, capacities of quantum channels~\cite{NH09, Kaur2018, BBFS18}, Bell's inequalities~\cite{TDS03, KGM17}, and other information-theoretic scenarios~\cite{Lan16, VV-Markov}. More broadly, the extendibility problem is a special case of the QMA-complete quantum marginal problem~\cite{Klyachko_2006, tyc, CLL13, Schilling14, Tyc2015, Liu06, LCV07}, which has been referred to in quantum chemistry as the $N$-representability problem~\cite{Tred57,Coul60,Cole63}. For fixed $k$, the extendibility problem can be formulated as a semidefinite program (SDP), making it efficiently solvable for low-dimensional systems $A$ and $B$~\cite{DPS02, complete-extendibility}. Analytic conditions for $k$-extendibility in finite-dimensional systems are known only for particular values of $k$ and/or for special classes of states~\cite{Ranade09, JV13, CJKLZ14, FLSC14, KGM17}.

In the infinite-dimensional case, of central relevance for quantum-optical realizations, the theory of Gaussian entanglement has been explored thoroughly in the past two decades~\cite{adesso14, BUCCO, revisited}. However, more general extendibility questions have been approached sparingly. The only work that we are aware of is~\cite{Bhat16}, where it was shown that a Gaussian state is separable if and only if it is Gaussian $k$-extendible for all $k$.

Here we study and characterize the full hierarchy of extendibility for quantum Gaussian states. After showing that any Gaussian state is $k$-extendible if and only if it is Gaussian $k$-extendible, we derive a simple SDP in terms of the state's covariance matrix in order to decide its $k$-extendibility. The size of our SDP scales linearly with the number of local modes. We also provide an analytic condition that completely characterizes the set of $k$-extendible states in the case of the extended system containing one mode only, generalizing the well-known positive partial transpose (PPT) criterion~\cite{PeresPPT, HorodeckiPPT, Simon00}.
We then discuss several applications of this result, deriving along the way: (i) analytic conditions for $k$-extendibility for all single-mode Gaussian channels; (ii) a tight de Finetti-type theorem bounding the distance between any $k$-extendible Gaussian state and the set of separable states; tight upper bounds on (iii) R\'{e}nyi relative entropy of entanglement and (iv) R\'{e}nyi entanglement of formation for any $k$-extendible Gaussian state. Our results reach unexplored depths in 
the ocean of continuous-variable quantum information.

\paragraph{Gaussian states.} We recall the basic theory of quantum Gaussian states~\cite{WANG20071, BUCCO, adesso14, weedbrook12}. Let $x_j$ and $p_j$ ($1\leq j\leq n$) denote the canonical operators of a system of $n$ harmonic oscillators (modes), arranged as a vector $r\coloneqq (x_1,p_1, \ldots, x_n, p_n)^\t$. The canonical commutation relations can be compactly written as $[r,r^\t ] = \I \Omega$, where $\Omega\coloneqq \lsmatrix 0 & 1 \\ -1 & 0 \rsmatrix^{\oplus n}$ is the standard symplectic form. Given any (not necessarily Gaussian) $n$-mode state $\rho$, its mean or displacement vector is $s \coloneqq \Tr[r \, \rho] \in \mathds{R}^{2n}$, while its quantum covariance matrix (QCM) is the $2n\times 2n$ real symmetric matrix $V\coloneqq \Tr \left[ \{r-s,(r-s)^\t\} \, \rho\right]$. Gaussian states $\rho^\G$ are (limits of) thermal states of quadratic Hamiltonians and are uniquely identified by their displacement vector $s$ and QCM $V$. We shall often assume $s=0$, since the mean can be adjusted by local displacement unitaries that do not affect $k$-extendibility.
Physically legitimate QCMs $V$ satisfy the Robertson--Schr\"odinger uncertainty principle $V \geq \I \Omega$, hereafter referred to as the bona fide condition~\cite{simon94}. Any matrix obeying this condition can be the QCM of a Gaussian state.

\paragraph{Extendibility of Gaussian states.} Let $\rho_{AB}$ be a (not necessarily Gaussian) state of a bipartite system of $n=n_A+n_B$ modes. We assume that $\rho_{AB}$ has vanishing first moments and finite second moments, so that we can construct its QCM
\begin{equation}\label{V AB}
	V_{AB} = \begin{pmatrix} V_A & X \\ X^{\t} & V_B \end{pmatrix} .
\end{equation}
It can be shown \footnote{See the Supplemental Material for proofs and additional technical derivations.} that every $k$-extension $\widetilde{\rho}_{AB_1\ldots B_k}$ of $\rho_{AB}$ also has (a) vanishing first moments and (b) finite second moments, arranged in a QCM of the form
\begin{equation} \label{V AB1...Bk}
	\widetilde{V}_{AB_1\ldots B_k} = \begin{pmatrix} V_A & X & X & \ldots & X \\ X^{\t} & V_B & Y  & \ldots & Y \\[-1ex] X^{\t} & Y & V_B & \ddots & \vdots \\[-1ex] \vdots & \vdots & \ddots & \ddots & Y \\ X^{\t} & Y & \ldots & Y & V_B \end{pmatrix} ,
\end{equation}
where $Y$ is a symmetric matrix. A similar structure had already been identified in~\cite{Bhat16}; however, there the crucial fact that $Y$ needs to be symmetric was not observed. We are now concerned with the $k$-extendibility of Gaussian states. Our first result indicates that Gaussian states are in some sense a closed set under $k$-extensions:
\begin{theorem}\label{GClosed}
A Gaussian state $\rho_{AB}^\G$ is $k$-extendible if and only if it has a Gaussian $k$-extension. 
\end{theorem}
\begin{proof}
Let $\widetilde{\rho}_{AB_1\ldots B_k}$ be a (not necessarily Gaussian) $k$-extension of $\rho_{AB}^\G$. Consider $m$ identical copies of it across the systems $A_\ell B_{\ell 1}\ldots B_{\ell k}$, where $1\leq \ell\leq m$. For $1\leq j\leq k$, let $U_j$ be a passive unitary that acts on the annihilation operators $b_{\ell j}$ of the systems $B_{\ell j}$ so that $U_j^\dag b_{1j} U_j=\frac{b_{1j}+\ldots + b_{mj}}{\sqrt{m}}$. Set
\begin{equation}
\omega^{(m)}_{A_{1}B_{11} \ldots B_{mk}} \coloneqq (U_1\otimes \ldots \otimes U_k) \left( \bigotimes_{\ell=1}^m \widetilde{\rho}_{A_\ell B_{\ell 1}\ldots B_{\ell k}}\right) (U_1\otimes \ldots \otimes U_k)^\dag\, .
\label{omega extended}
\end{equation}
By the quantum central limit theorem~\cite{Cushen1971, Petz1992}, the reduced state $\omega^{(m)}_{A_1B_{11}\ldots B_{1k}}$ satisfies 
$\lim_{m\to\infty} \left\|\omega^{(m)}_{A_1B_{11}\ldots B_{1k}} - \widetilde{\rho}^{\,\G}_{AB_1\ldots B_k}\right\|_1 = 0$,
where $\widetilde{\rho}^{\,\G}_{AB_1\ldots B_k}$ is the Gaussian state with the same first and second moments as $\widetilde{\rho}_{AB_1\ldots B_k}$, and $A_1\equiv A$, $B_{1j}\equiv B_j$~\cite{Note1}.

We now show that $\widetilde{\rho}^{\,\G}_{AB_1\ldots B_k}$ is indeed a Gaussian $k$-extension of $\rho_{AB}^\G$. First, it is symmetric under the exchange of any two $B$ systems, say $B_1\leftrightarrow B_2$. In fact, (i) the state in~\eqref{omega extended} is invariant under the exchange $(B_{11} , \ldots , B_{m1})\leftrightarrow (B_{12}, \ldots , B_{m2})$; (ii) consequently, the reduced state $\omega^{(m)}_{A_1B_{11}\ldots B_{1k}}$ is invariant under the exchange $B_{11}\leftrightarrow B_{12}$; (iii) symmetry is preserved under limits. Finally, to show that $\widetilde{\rho}^{\,\G}_{AB_1} = \rho^\G_{AB}$ under the identification $B_1\equiv B$, we observe that the QCM of $\widetilde{\rho}^{\,\G}_{AB_1\ldots B_k}$, which is the same as that of $\widetilde{\rho}_{AB_1\ldots B_k}$, is as in~\eqref{V AB1...Bk}. Since its upper-left $2\times 2$ corner corresponds to the QCM of $\rho^\G_{AB}$, we conclude that $\widetilde{\rho}^{\,\G}_{AB_1}$ and $\rho^\G_{AB}$ have the same first and second moments; being Gaussian, they must coincide.
\end{proof}

By virtue of Theorem~\ref{GClosed}, we can confine the search of $k$-extensions of Gaussian states to the same Gaussian realm. The next result shows that this reduces to an efficiently solvable SDP feasibility problem, with the size of the SDP scaling linearly in the number of modes of the $B$ system. In the case of $B$ being composed of one mode only, we find an analytic solution in the form of a simple necessary and sufficient condition for $k$-extendibility.

\begin{theorem}\label{G k ext thm}
	Let $\rho_{AB}$ be a $k$-extendible (not necessarily Gaussian) state of $n_A+n_B$ modes with QCM $V_{AB}$. Then there exists a $2n_B\times 2n_B$ quantum covariance matrix $\Delta_B\geq\I\Omega_B$ such that 
	\begin{equation}\label{k ext necessary 2}
		V_{AB}\geq \I\Omega_A\oplus \left(\left(1-\frac1k\right)\Delta_B+\frac1k\I\Omega_B\right).
	\end{equation}
	Moreover, the above condition is necessary and sufficient for $k$-extendibility when $\rho_{AB}=\rho_{AB}^\G$ is Gaussian. If in addition $n_B=1$, then $\rho_{AB}^\G$ is $k$-extendible if and only if
	\begin{equation}\label{simplified k ext necessary 2}
		V_{AB}\geq\I\Omega_A\oplus\left(-\left(1-\frac2k\right)\I\Omega_B\right) .
	\end{equation}
\end{theorem}

In the proof of Theorem~\ref{G k ext thm}, we employ the following characterization of positive semidefiniteness of Hermitian block matrices~\cite[Theorem~1.12]{ZHANG1}:
\begin{equation}\label{eq-Schur_complement}
	\begin{aligned}
	M\!=\!\begin{pmatrix} P & Z \\ Z^\dag & Q\end{pmatrix} \geq 0\ \Leftrightarrow\ P \geq 0 ,\ M/P\coloneqq Q - Z^\dag P^{-1} Z \geq 0\, ,
	\end{aligned}
\end{equation}
where the matrix $M/P$ is called the \emph{Schur complement} of $M$ with respect to $P$. For details concerning the degenerate case of non-invertible $P$, see~\cite{Note1}. Using~\eqref{eq-Schur_complement}, for any QCM $V_{AB}$ as in~\eqref{V AB}, the inequality in~\eqref{k ext necessary 2} and the condition $\Delta_B\geq\I\Omega_B$ can be written together as
\begin{equation}\label{k ext necessary 1}
	\I\Omega_B\leq \Delta_B \leq \frac{k}{k-1} \left( V_B - X^{\t} (V_A-\I\Omega_A)^{-1} X \right) - \frac{1}{k-1}\, \I\Omega_B\, .
\end{equation}
Analogously,~\eqref{simplified k ext necessary 2} is equivalent to
\begin{equation}\label{simplified k ext necessary 1}
	V_B - X^{\t} (V_A - \I\Omega_A)^{-1} X  \geq - \left( 1- \frac2k\right) \I\Omega_B\, .
\end{equation}

\begin{proof}[Proof of Theorem~\ref{G k ext thm}]
We first establish necessity of~\eqref{k ext necessary 2} for $k$-extendibility of an arbitrary state $\rho_{AB}$.
If $\rho_{AB}$ is $k$-extendible, then there exists a matrix $\widetilde{V}_{AB_1\ldots B_k}$ as in~\eqref{V AB1...Bk} that obeys the bona fide condition $\widetilde{V}_{AB_1\ldots B_k} \geq \I \big(\Omega_A \oplus \Omega_{B_1\ldots B_k}\big)$. Using~\eqref{eq-Schur_complement}, and noting that $V_A\geq \I \Omega_A$ holds  because $\rho_A$ is a valid state, we arrive at the inequality $\big(\widetilde{V}_{AB_1\ldots B_k} - \I\Omega_A\big)\big/ \big(\widetilde{V}_A - \I\Omega_A\big) \geq \I\Omega_{B_1\ldots B_k}$. Using~\eqref{V AB1...Bk}, and letting $\ket{+}\coloneqq \frac{1}{\sqrt{k}}\sum_{j=1}^k \ket{j}\in\mathds{R}^k$, upon elementary manipulations this can be rephrased as
\begin{equation*}
\begin{aligned}
&(\mathds{1}_k - \ketbra{+}) \otimes (V_B - Y-\I \Omega_B) \\
&+ \ketbra{+} \otimes \left( V_B + (k-1) Y - k X^{\t} \left( V_A - \I \Omega_A\right)^{-1} X - \I \Omega_B \right) \geq 0\, .
\end{aligned}
\end{equation*}
Since the first factors of the above two addends are orthogonal to each other, positive semidefiniteness can be imposed separately on the second factors. Letting $\Delta_B\coloneqq V_B-Y$, we obtain~\eqref{k ext necessary 1}, whose equivalence to~\eqref{k ext necessary 2} follows by applying~\eqref{eq-Schur_complement}. To deduce~\eqref{simplified k ext necessary 2} from~\eqref{k ext necessary 2}, simply substitute the complex conjugate bona fide condition $\Delta_B \geq -\I\Omega_B$ into~\eqref{k ext necessary 2}.

By Theorem~\ref{GClosed}, the condition $\widetilde{V}_{AB_1\ldots B_k} \geq \I \big(\Omega_A \oplus \Omega_{B_1\ldots B_k}\big)$ is also sufficient to ensure $k$-extendibility when $\rho_{AB}=\rho_{AB}^\G$ is Gaussian. By the above reduction, this condition is equivalent to that in~\eqref{k ext necessary 2}.

We now prove that when $n_B=1$,~\eqref{simplified k ext necessary 2} implies the existence of a real $\Delta_B$ such that~\eqref{k ext necessary 1} is satisfied. By~\cite[Lemma~7]{revisited}, we know that~\eqref{k ext necessary 1} is satisfied for some real $\Delta_B$ if and only if
\begin{equation}
	\frac{k}{k-1} \left( V_B - X^{\t} (V_A-\I\Omega_A)^{-1} X \right) - \frac{1}{k-1}\, \I\Omega_B \geq \pm \I \Omega_B\, ,
\end{equation}
meaning that both inequalities are satisfied. Using~\eqref{eq-Schur_complement}, we see that the condition with the $+$ reduces to $V_{AB}\geq \I\Omega_{AB}$, which is guaranteed to hold by hypothesis. That with the $-$ yields instead~\eqref{simplified k ext necessary 1}, which is in turn equivalent to~\eqref{simplified k ext necessary 2}.
\end{proof}

Although some of the above manipulations formally resemble those in~\cite{Bhat16}, the two arguments are conceptually different and lead to different conclusions~\cite{Note1}: in fact, in~\cite{Bhat16}, the question of $k$-extendibility of Gaussian states is explicitly mentioned as an outstanding problem.

Recall that a bipartite state is separable if and only if it is $k$-extendible for all $k$~\cite{FLV88, RW89, DPS02, complete-extendibility} and that any $k$-extendible state is also $(k-1)$-extendible. Thus, taking the limit $k\to\infty$ of condition~\eqref{k ext necessary 2} shows that $\rho_{AB}^\G$ is separable if and only if there exists a $2n_B\times 2n_B$ matrix $\Delta_B\geq\I\Omega_B$ such that $V_{AB}\geq \I\Omega_A\oplus \Delta_B$.
This reproduces the analytic condition for separability of Gaussian states found in~\cite[Theorem~5]{revisited}. In the same limit $k\to \infty$, it is also easy to verify that condition~\eqref{simplified k ext necessary 2} reduces to the PPT criterion~\cite{PeresPPT, HorodeckiPPT, Simon00, Werner01, revisited}.


It turns out that the necessary condition in~\eqref{simplified k ext necessary 2} is no longer sufficient when $n_B>1$. This is demonstrated by the example of the $(2+2)$-mode bound entangled Gaussian state constructed in~\cite{Werner01}, which  obeys~\eqref{simplified k ext necessary 2} for all $k$ (because it is PPT) yet it is not even $2$-extendible~\cite{Note1}.

Theorem~\ref{G k ext thm} also reveals an implication of $2$-extendibility for Gaussian {\em steerability}, i.e., Einstein--Podolsky--Rosen steerability via Gaussian measurements~\cite{wise, steerability, Simon16, Kor, Lami16}. The $k=2$ case of~\eqref{simplified k ext necessary 2} shows that any Gaussian state that is 2-extendible on $B$ is necessarily $B\to A$ Gaussian unsteerable, and hence useless for one-sided-device-independent quantum key distribution. When $n_B=1$, this condition is also sufficient, i.e., $2$-extendibility is equivalent to $B\to A$ Gaussian unsteerability.

\paragraph{Extendibility of Gaussian channels.} We now apply Theorem~\ref{G k ext thm} to study $k$-extendibility of single-sender single-receiver Gaussian quantum channels. A quantum channel $\mathcal{N}_{A\to B}$ is called $k$-extendible~\cite{Pankowski2013, Kaur2018} if there exists another quantum channel $\widetilde{\mathcal{N}}_{A\to B_1\dotsb B_k}$ from the sender $A$ to $k$ receivers $B_1,\dotsc, B_k$ such that the reduced channel from the sender to any one of the receivers is the same as the original channel $\mathcal{N}_{A\to B}$.

A Gaussian channel $\mathcal{N}_{A\to B}$ with $n$ input modes and $m$ output modes maps Gaussian states to Gaussian states and is uniquely characterized by a real $2m\times 2n$ matrix $X$, a real symmetric $2m\times 2m$ matrix $Y$, and a real vector $\delta\in \mathds{R}^{2m}$, such that $Y+\I\Omega\geq \I X\Omega X^T$~\cite{BUCCO}. Its action can be described directly in terms of the mean vector $s$ and QCM $V$ of the input Gaussian state as follows: $s\mapsto Xs+\delta$, $V\mapsto XVX^{\t}+Y$. In what follows, we set $\delta=0$ without loss of generality.	

To any channel $\mathcal{N}_{A\to B}$ we can associate its Choi--Jamio\l kowski state $\rho_{AB}^{\mathcal{N}}(r)\coloneqq\mathcal{N}_{A'\to B}\left(\ketbra{\psi_r}^{\otimes n}_{AA'}\right)$, where for~$r>0$ the two-mode squeezed vacuum is defined as $\ket{\psi_r} \coloneqq \mathrm{sech}(r) \sum_{j=0}^{\infty}\tanh(r)^j\ket{j,j}$~\cite{Holevo-CJ}. It can be seen that $\mathcal{N}_{A\to B}$ is $k$-extendible if and only if $\rho_{AB}^{\mathcal{N}}(r)$ is $k$-extendible on $B$ for some (and hence all) $r>0$~\cite{Note1}. The same conclusion follows from arguments in~\cite{nogo3,Wolf2007,NFC09}. For any Gaussian channel $\mathcal{N}$, the state $\rho_{AB}^{\mathcal{N}}(r)$ is Gaussian. Hence, via Theorem~\ref{G k ext thm}, we deduce that a Gaussian channel is $k$-extendible if and only if there exists a $2m\times 2m$ real matrix $\Delta$ such that
\begin{equation}\label{k ext Choi state}
	\I\Omega\leq \Delta \leq \frac{k}{k-1} \left( Y + \I X \Omega X^{\t} \right) - \frac{1}{k-1} \I\Omega\, .
\end{equation}
When $m=1$, this is equivalent to 
$	Y + \I X\Omega X^{\t} + \left( 1-2/k\right) \I \Omega \geq 0$.
If also $n=1=m$, a simplified equivalent condition that incorporates also the complete positivity requirements is
\begin{equation}\label{eq:single-mode-k-ext-condition}
	\sqrt{\det Y} \geq 1 - \frac1k + \left| \det X - \frac1k\right| .
\end{equation}

By applying~\eqref{eq:single-mode-k-ext-condition}, we find necessary and sufficient conditions for the $k$-extendibility of all possible single-mode Gaussian channels, which play a prominent role in modelling optical quantum communication~\cite{Holevo2007,EW07,BUCCO}. By the results of~\cite{Holevo2007}, the following characterization of $k$-extendibility for three fundamental single-mode Gaussian channels suffices to solve the problem for \emph{all} single-mode Gaussian channels~\cite{Note1}:

(i) The thermal channel of transmissivity $\eta\in(0,1)$ and environment thermal photon number $N_B \geq 0$ is defined by $X=\sqrt{\eta}\mathds{1}$ and $Y=(1-\eta)(2N_B +1)\mathds{1}$. It is $k$-extendible if and only if $\eta \leq \frac{N_B+1/k}{N_B+1}$. For the case $N_B=0$, corresponding to a pure-loss channel, this reduces to $\eta \leq 1/k$.

(ii) The amplifier channel of gain $G>1$ and environment thermal photon number $N_B \geq 0$ is defined by $X=\sqrt{G} \mathds{1}$ and $Y=(G-1)(2N_B+1) \mathds{1}$. This channel is $k$-extendible if and only if $N_B>0$ and $G\geq \frac{N_B+1-1/k}{N_B}$.

(iii) The additive noise channel with noise parameter $\xi>0$ is defined by $X=\mathds{1}$ and $Y = \xi \mathds{1}$. This channel is $k$-extendible if and only if $\xi \geq 2\left(1-1/k\right)$.

As expected, the above conditions reduce to their entanglement-breaking counterparts from~\cite{Holevo2008} for $k\to \infty$.

\paragraph{Distance between $k$-extendible and separable states.} A problem of central interest in quantum information theory is determining how close $k$-extendible states are to the set of separable states. In~\cite[Theorem~II.7']{1-1/2-de-Finetti}, it was found that a finite-dimensional $k$-extendible state is $4d^2/k$-close to the set of separable states in trace norm, where $d$ is the dimension of the extended system. Moreover, it was also shown~\cite[Corollary~III.9]{1-1/2-de-Finetti} that the error term in the approximation necessarily depends on $d$ at least linearly. One can instead obtain a $\ln d$ dependence by resorting to different norms~\cite{faithful}.

Can similar estimates be provided in the Gaussian case? Results in this setting have been obtained in~\cite{Koenig2009} for fully symmetric systems of the form $B_1\ldots B_k$. Here we extend these de Finetti theorems to the case where the symmetry is relative to a fixed reference system $A$. We are interested in the distance of a given Gaussian state $\rho_{AB}^\G$ to the set $\text{SEP}(A\!\!:\!\!B)$ of bipartite separable states on systems $A$ and $B$, as measured by either (i) the trace norm, yielding the quantity $\left\| \rho_{AB}^\G-\text{SEP}(A\!:\!B)\right\|_1\coloneqq \inf_{\sigma_{AB}\in\text{SEP}(A:B)} \left\| \rho_{AB}-\sigma_{AB}\right\|_1$, or (ii) the quantum Petz--R\'enyi relative entropy $D_\alpha(\rho\Vert \sigma)\coloneqq \frac{1}{\alpha-1}\ln \Tr[\rho^\alpha \sigma^{1-\alpha}]$ for $\alpha>0$~\cite{PetzRenyi}, which leads to the measure $E_{R,\alpha}(\rho^\G_{AB})\coloneqq \inf_{\sigma_{AB}\in\text{SEP}(A:B)} D_\alpha(\rho_{AB}\Vert\sigma_{AB})$. For $\alpha=1$ the Petz--R\'enyi relative entropy reduces to the Umegaki relative entropy~\cite{U62}, and we obtain the standard relative entropy of entanglement~\cite{VP98}.
We find the following:
\begin{theorem}\label{thm-dist_sep_k_ext}
	Let $\rho_{AB}^\G$ be a $k$-extendible Gaussian state of $n\coloneqq n_A+n_B$ modes. Then,
	\begin{align} \nonumber \\[-5ex]
		\left\| \rho_{AB}^\G-\operatorname{SEP}(A\!:\!B)\right\|_1 &\leq \frac{2n}{k}\, , \label{distance k ext} \\
		E_{R,\alpha}(\rho_{AB}^\G) \leq n\, \ln\left(1+\frac{\eta_{k,\alpha}}{k-1}\right) &\leq \frac{n\, \eta_{k,\alpha}}{k-1}\, , \label{E R k ext}
	\end{align}
	where $\eta_{k,\alpha}=1$ if $\alpha\leq k+1$, and $\eta_{k,\alpha}=2$ otherwise.
\end{theorem}

The proof is in~\cite{Note1}. Remarkably, the upper bounds in~\eqref{distance k ext}--\eqref{E R k ext} hold universally for all Gaussian states, independently, e.g., of their mean photon number. This is in analogy with the main results of~\cite{Koenig2009}, and constitutes a \emph{quantitative} improvement over the finite-dimensional case, where---as we mentioned before---the bound has to depend on the underlying dimension. Furthermore, for two-mode states, the bounds in~\eqref{distance k ext}--\eqref{E R k ext} can be shown to be tight up to a constant for all $k$ and all $\alpha\geq 1$. Namely, for all $k\geq 2$ there exists a $k$-extendible two-mode Gaussian state $\rho_{AB}^\G$ such that $\left\| \rho_{AB}^\G-\text{SEP}(A\!:\!B)\right\|_1\geq \frac{1}{2k-1}$ and $E_{R,1}(\rho_{AB}^\G)\geq E_D(\rho_{AB}^\G)\geq \ln \frac{k}{k-1} - o(1)$ as $r\to\infty$, where $E_D$ denotes the distillable entanglement~\cite{Note1}.

\paragraph{Entanglement of formation of Gaussian $k$-extendible states.} We now show that one can also obtain an upper bound on the entanglement of formation of Gaussian $k$-extendible states. This is a \emph{qualitative} improvement over the finite-dimensional case, as a result of this kind has no analogue in that setting.
We employ the recently developed theory of R\'enyi-2 Gaussian correlation quantifiers~\cite{AdessoSerafini, steerability, Lami16, LL-log-det}, and especially the monogamy of the Gaussian R\'enyi-2 version of the entanglement of formation~\cite{Lami16}, which stems in turn from the equality between this measure and the Gaussian R\'enyi-2 squashed entanglement~\cite{LL-log-det}.

For a bipartite state $\rho_{AB}$ and for some $\alpha\geq 1$, the \textit{R\'enyi-$\alpha$ entanglement of formation} $E_{F,\alpha}(\rho_{AB})$ is defined as the infimum of $\sum_i p_i \, S_{\alpha}\big(\psi_{A}^{(i)}\big)$ over all pure-state decompositions $\sum_i p_i \psi_{AB}^{(i)}=\rho_{AB}$ of $\rho_{AB}$~\cite{Horodecki-review}. Here, $S_\alpha(\sigma)\coloneqq \frac{1}{1-\alpha}\ln \Tr[\sigma^\alpha]$ is the R\'enyi-$\alpha$ entropy.

For a Gaussian state $\rho_{AB}^\G$ with QCM $V_{AB}$, we can derive an upper bound on $E_{F,\alpha}(\rho_{AB}^\G)$ by restricting the decompositions to include pure Gaussian states only. This leads to the \textit{Gaussian R\'enyi-${\alpha}$ entanglement of formation}, given by~\cite{Wolf03, LL-log-det}
\begin{equation}\label{GEoF}
	E^{\G}_{F,\alpha} \left( \rho_{AB}^\G\right) = \inf \big\{ S_\alpha (\gamma_A):\ \text{$\gamma_{AB}$ pure QCM and $\gamma_{AB}\leq V_{AB}$} \big\}\, ,
\end{equation}
where we denote by $S_{\alpha}(W)$ the R\'enyi-$\alpha$ entropy of a Gaussian state with QCM $W$, and ``pure'' QCMs are those that correspond to pure Gaussian states. 
While the typical choice $\alpha=1$ yields the standard entanglement of formation, R\'enyi-$2$ quantifiers arise naturally in the Gaussian setting, as they reproduce Shannon entropies of measurement outcomes~\cite{AdessoSerafini,LL-log-det}. For $\alpha=2$, Eq.~\eqref{GEoF} becomes
\begin{equation}
E^{\G}_{F,2} \left( \rho_{AB}\right) = \min \left\{ M(\gamma_A): \text{$\gamma_{AB}$ pure QCM and $\gamma_{AB}\leq V_{AB}$} \right\},
\end{equation}
where for a positive definite matrix $V$ we set $M(V) \coloneqq S_2(V) = \frac12 \ln \det V$. We then find the following:

\begin{theorem}\label{thm-Gaussian_EoF}
	The R\'enyi-2 Gaussian entanglement of formation of a $k$-extendible Gaussian state $\rho_{AB}^\G$ of $n_A+n_B$ modes with QCM $V_{AB}$ is bounded from above as
$		E^{\G}_{F,2}\left(\rho_{AB}^\G \right) \leq \frac{M(V_A)}{k}\, .$
	Consequently, the standard entanglement of formation of $\rho_{AB}^\G$  satisfies
$		E_{F,1}\left(\rho_{AB}^\G\right) \leq E_{F,1}^{\G}\left(\rho_{AB}^\G\right) \leq n_A \, \varphi\left(\frac{M(V_A)}{n_A k}\right)$,
	where $\varphi(x)\coloneqq \frac{\e^x+1}{2}\ln \left(\frac{\e^x+1}{2}\right)-\frac{\e^x-1}{2}\ln \left(\frac{\e^x-1}{2}\right)$.
\end{theorem}

The function $M$ plays the role of some ``effective dimension'' in the bounds above. It is related to other quantities conventionally thought of as infinite-dimensional substitutes for the dimension, such as the mean photon number, defined for a state $\rho$ of $n$ modes as $\left<N\right> = \left<N\right>_\rho \coloneqq \Tr\left[  \left( \sum_j a_j^\dag a_j \right) \rho \right]$. When $\rho$ is zero-mean Gaussian and has QCM~$V$, one has $\left<N\right> = \frac{1}{4} \left( \Tr V - 2n \right)$. By applying the arithmetic--geometric mean inequality, one can show that $M(V) \leq n \ln \left(\frac{2\left<N\right>}{n} + 1\right)$, which can be further relaxed to $M(V)\leq 2\left<N\right>$.

\paragraph{Summary \& outlook.} We accomplished a comprehensive analysis of the $k$-extendibility of Gaussian quantum states. We determined that a Gaussian state is $k$-extendible if and only if it is Gaussian $k$-extendible, which allowed us to derive a simple semidefinite program that solves the problem completely in a computationally efficient way. When the extended system contains one mode only, we fully characterized the set of $k$-extendible Gaussian states by a simple analytic condition reminiscent of the PPT criterion. We demonstrated further applications to Gaussian state steerability, $k$-extendiblity of Gaussian channels, bounding the distance between $k$-extendible and separable states, and the R\'{e}nyi entanglement of formation for Gaussian states. Our results also yield necessary criteria for $k$-extendibility of non-Gaussian states based on second moments. This work sheds novel light onto the fine structure of entanglement and its uses in continuous-variable systems.

It remains an intriguing open problem to find an analytic condition for $k$-extendibility of arbitrary Gaussian states. 
Another topic for future work is to explore applications of Theorem~\ref{G k ext thm} to the non-asymptotic capacities of Gaussian channels, in light of recent work~\cite{Kaur2018,BBFS18} exploiting $k$-extendibility to bound the performance of quantum processors.


\begin{acknowledgments}

\paragraph{Acknowledgments.} LL and GA acknowledge financial support from the European Research Council under the Starting Grant GQCOP (Grant No.~637352).
SK and MMW acknowledge support from the NSF under Grant No.~1714215. S.K. acknowledges support from the NSERC PGS-D.

\end{acknowledgments}

\bibliography{biblio}{}

\clearpage

\onecolumngrid
\begin{center}
\vspace*{\baselineskip}
{\textbf{\large Supplemental Material}}\\
\end{center}

\renewcommand{\theequation}{S\arabic{equation}}
\setcounter{equation}{0}
\setcounter{figure}{0}
\setcounter{table}{0}
\setcounter{section}{0}
\setcounter{page}{1}
\makeatletter

\section{Background}

For a continuous-variable system of $n$ modes, we let $r=(x_1,p_1,\dotsc,x_n,p_n)^{\t}$ denote the vector of position- and momentum- quadrature operators. These operators satisfy the canonical commutation relations $[x_j,p_k]=\I\delta_{j,k}\mathds{1}$ for all $1\leq j,k\leq n$, which we can rewrite compactly as
\begin{equation}
    [r,r^\t] = \I \Omega = \I \begin{pmatrix} 0 & 1 \\ -1 & 0 \end{pmatrix}^{\oplus n} .
    \label{CCR}
\end{equation}
Given a quantum state represented by a density matrix $\rho$, we can construct the real vector $s\in\mathds{R}^{2n}$ defined by $s_j=\Tr[r_j\, \rho ]$ for all $1\leq j\leq 2n$, called the mean vector of $\rho$, and the real symmetric $2n\times 2n$ matrix $V$ satisfying $V_{ij}=\Tr[ \{r_i-s_i, r_j-s_j\} \,\rho]$, called the quantum covariance matrix (QCM) of $\rho$, where $\{A,B\}\coloneqq AB+BA$ is the anticommutator. The QCM of every quantum state $\rho$ necessarily satisfies a Robertson--Schr\"odinger uncertainty principle of the form~\cite{simon94}
\begin{equation}
    V\geq \I \Omega
    \label{Heisenberg}.
\end{equation}
The QCM $V_{AB}$ of any bipartite state $\rho_{AB}$ has the block matrix form
\begin{equation}\label{eq-V_AB_supp}
	V_{AB}=\begin{pmatrix} V_A & X \\ X^{\t} & V_B \end{pmatrix},
\end{equation}
where $V_A$ and $V_B$ are the QCMs of the reduced states $\rho_A$ and $\rho_B$, respectively, and $X$ describes the correlations between the systems $A$ and $B$.

A state $\rho$ on $n$ modes is called a \textit{Gaussian} state if it is either a thermal state of some quadratic Hamiltonian, i.e., if it can be written in the form $\rho=\frac{\e^{-\beta H(A,x)}}{\Tr[\e^{-\beta H(A,x)}]}$ for some $\beta > 0$, $x\in\mathds{R}^{2n}$, and a real symmetric $2n\times 2n$ matrix $A$, where $H(A,x)=\frac{1}{2}r^{\t}Ar+r^{\t}x$ is quadratic in the canonical operators, or it is a limit of states of that form.

Gaussian states are uniquely described by their mean vector and quantum covariance matrix. Furthermore, any Gaussian state $\rho$ on $n$ modes with QCM $V$ can always be brought into a canonical form by means of a symplectic unitary $U_S$, which acts on the mode operators $r$ as $U_{\!S}^{\mathstrut} r\, U_{\!S}^\dag = Sr$, where $S$ is a symplectic matrix, i.e., a matrix satisfying the defining relation $S\Omega S^{\t}=\Omega$. If the Williamson canonical form of $V$ is
\begin{equation}\label{V Williamson}
	V = S \left( \bigoplus_{j=1}^n \nu_j \mathds{1}_2\right) S^{\t},
\end{equation}
one has
\begin{equation} \label{rho canonical}
	\rho = U_S^\dag \left( \bigotimes_{j=1}^n \rho^\G(\nu_j) \right) U_S\, ,
\end{equation}
where the canonical form of one-mode Gaussian states is defined in the Fock basis by
\begin{equation}\label{1 mode canonical}
	\rho^\G(\nu)\coloneqq \frac{2}{\nu+1} \sum_{\ell=0}^\infty \left(\frac{\nu-1}{\nu+1}\right)^\ell \ketbra{\ell}.
\end{equation}

For $\lambda\in [0,1]$ we denote by $\mathcal{L}_\lambda$ the \textit{attenuator channel} of parameter $\lambda$, defined by
\begin{equation}\label{attenuator}
	\mathcal{L}_\lambda(\cdot) \coloneqq \Tr_B \left[ U_\lambda ((\cdot)_A \otimes{|0\rangle\langle 0|}_B) U_\lambda^\dag \right],
\end{equation}
where $U_\lambda$ is the symplectic unitary that implements a beam splitter with transmissivity $\lambda$, and on two modes---whose annihilation operators we denote by $a,b$---takes the form
\begin{equation}\label{beam splitter}
	U_\lambda \coloneqq e^{-\arccos\left(\!\sqrt{\lambda}\right)\,(ab^\dag - a^\dag b)} = e^{-\sqrt{\frac{1-\lambda}{\lambda}} ab^\dag} \lambda^{\frac{a^\dag a - b^\dag b}{2}} e^{\sqrt{\frac{1-\lambda}{\lambda}} a^\dag b},
\end{equation}
where we use the function $\arccos:[0,1]\to [0,\pi/2]$, and the last identity is a rewriting of~\cite[Eq.~(5.116)]{BUCCO}. Upon tedious but straightforward algebraic manipulations,~\eqref{attenuator} and~\eqref{beam splitter} together yield the following Kraus representation of the attenuator channel~\cite[Eq.~(4.5)]{Ivan2011}:
\begin{equation}\label{attenuator Kraus}
	\mathcal{L}_\lambda(\cdot) = \sum_{j=0}^\infty \frac{\left(\frac{1}{\lambda} -1\right)^j}{j!}\ a^j \lambda^{a^\dag a/2} (\cdot) \lambda^{a^\dag a/2} (a^\dag)^j\, .
\end{equation}

\section{Extensions of states of continuous-variable systems}

Throughout this section we clarify some subtleties related to extensions in continuous-variable systems. We start by asking whether states with bounded energy are in some sense a closed set under $k$-extensions. The reason why this is important is because those states are naturally the most physically relevant.

\begin{lemma} \label{simple lemma}
Let $\rho_{AB}$ be a (not necessarily Gaussian) state with vanishing first moments and finite second moments, identified by a QCM $V_{AB}$ as in~\eqref{V AB}. Then any $k$-extension $\widetilde{\rho}_{AB_1\ldots B_k}$ has (a) vanishing first moments and (b) finite second moments with a corresponding QCM of the form as in~\eqref{V AB1...Bk}.
\end{lemma}

\begin{proof}
We only show that the second moments must be finite, as the claims concerning the first moments are proved in an analogous fashion. To see this, it suffices to show that: (i) $\Tr\big[r_{A, \alpha}^2 \, \widetilde{\rho}_{AB_1\ldots B_k} \big] < \infty$ for all $1\leq \alpha\leq 2n_A$; and that (ii) $\Tr\big[r_{B_j,\,\beta}^2\, \widetilde{\rho}_{AB_1\ldots B_k}  \big] < \infty$ for all $1\leq j\leq k$ and $1\leq \beta\leq 2n_B$. Here, $r_{B_j,\,\beta}$ denotes the $\beta$-th component of the vector of canonical operators acting on $B_j$, and analogously for $r_{A,\alpha}$. For (ii) one writes
\begin{equation}
\Tr\left[r_{B_j,\,\beta}^2\, \widetilde{\rho}_{AB_1\ldots B_k} \right] = \Tr \left[ r_{B_j,\,\beta}^2 \, \widetilde{\rho}_{AB_j}  \right] = \Tr \left[ r_{B,\,\beta}^2 \, \rho_{AB}^\G  \right] < \infty\, .
\end{equation}
The proof of (i) follows the same lines.

Since the second moments of $\,\widetilde{\rho}_{AB_1\ldots B_k}$ have been shown to be finite, we can now form the corresponding quantum covariance matrix $\widetilde{V}_{AB_1\dotsb B_k}$. The particular structure in~\eqref{V AB1...Bk} results from the requirements imposed on $k$-extensions. One such requirement is $\widetilde{\rho}_{AB_j}=\rho_{AB}$ for all $1\leq j\leq k$, which implies that the first row of $\widetilde{V}_{AB_1\dotsb B_k}$ must feature identical copies of the matrix $X$. Indeed, for all $1\leq j\leq k$,
\begin{align}
    \left(\widetilde{V}_{AB_1\dotsb B_k}\right)_{A,\alpha\,;\, B_j,\,\beta} &= \Tr[(r_{A,\alpha}r_{B_j,\,\beta}+r_{B_j,\,\beta}r_{A,\alpha})\,\widetilde{\rho}_{AB_1\dotsb B_k}] \\
    &=\Tr[(r_{A,\alpha}r_{B_j,\,\beta}+r_{B_j,\,\beta}r_{A,\alpha})\,\widetilde{\rho}_{AB_j}]\\
    &=\Tr[(r_{A,\alpha}r_{B_j,\,\beta}+r_{B_j,\,\beta}r_{A,\alpha})\,\rho_{AB}]\\
    &=X_{\alpha,\,\beta}.
\end{align}
Symmetry with respect to the systems $B_1,\dotsc,B_k$ implies that $\widetilde{\rho}_{B_{j_1}B_{j_2}}=\widetilde{\rho}_{B_1B_2}$ for all $1\leq j_1,j_2\leq k$. Furthermore, all such two-party reduced states are invariant under swapping of the two subsystems. Therefore, the matrix $Y$ whose entries are given by
\begin{align}
    Y_{i,\ell}&\coloneqq \Tr[(r_{B_1,i}r_{B_2,\ell}+r_{B_2,\ell}r_{B_1,\ell})\widetilde{\rho}_{AB_1\dotsb B_k}]\\
    &=\Tr[(r_{B_1,i}r_{B_2,\ell}+r_{B_2,\ell}r_{B_1,\ell})\widetilde{\rho}_{B_1B_2}]
\end{align}
is symmetric: $Y_{i,\ell}=Y_{\ell,i}$ for all $1\leq i,\ell\leq 2n_B$. The covariance matrix $\widetilde{V}_{AB_1\dotsb B_k}$ thus has the structure as in~\eqref{V AB1...Bk}.
\end{proof}

Before we move on, we want to elaborate on the role of the symmetry requirement in Theorem~\ref{GClosed}. For the sake of this discussion, let us call a state $\rho_{AB}$ \emph{generally $k$-extendible} on the $B$ system if there exists a quantum state $\widetilde{\rho}_{AB_1\ldots B_k}$ on $A$ and $k$ copies of $B$ such that $\widetilde{\rho}_{AB_i}\equiv \rho_{AB}$ for all $i$, where $\widetilde{\rho}_{AB_i}$ refers to the reduced state over the systems $AB_i$, and for a fixed $i$ we identified $B_i\equiv B$. Unlike in the definition we adopt in the main text, here we are not requiring the whole state to be symmetric, only the reductions to always be the same. To stress the difference between the two cases, here and for the rest of this section we refer to standard $k$-extendible state as \emph{symmetrically $k$-extendible}.

Clearly, if $\rho_{AB}=\rho_{AB}^\G$ is Gaussian, we may want to look at Gaussian general $k$-extensions, in the same way as we did in the main text for symmetric $k$-extensions. Fortunately, it turns out that Theorem~\ref{GClosed} implies that  these concepts are  equivalent.

\begin{corollary} \label{GClosed stronger}
Let $\rho_{AB}^\G$ be a Gaussian state. Then the following are equivalent:
\begin{enumerate}
    \item[(i)] $\rho_{AB}^\G$ is generally $k$-extendible;
    \item[(ii)] $\rho_{AB}^\G$ is symmetrically $k$-extendible;
    \item[(iii)] $\rho_{AB}^\G$ has a Gaussian general $k$-extension;
    \item[(iv)] $\rho_{AB}^\G$ has a Gaussian symmetric $k$-extension.
\end{enumerate}
\end{corollary}

\begin{proof}
The general and well-known fact that $(i)\Leftrightarrow(ii)$ follows immediately by symmetrization arguments: given a general $k$-extension $\widetilde{\rho}_{AB_1\ldots B_k}$, we can construct the symmetric $k$-extension
\begin{equation}
    \widetilde{\rho}_{AB_1\ldots B_k}^{\,(\text{sym})}\coloneqq \frac{1}{k!}\sum_{\pi\in S_k} \left(\mathds{1}_A\otimes U_{B_1\ldots B_k}(\pi) \right) \widetilde{\rho}_{AB_1\ldots B_k} \left(\mathds{1}_A\otimes U_{B_1\ldots B_k}(\pi) \right)^\dag\, ,
\end{equation}
where $S_k$ denotes the symmetric group over $k$ elements, and $U_{B_1\ldots B_k}(\pi)$ is the unitary that permutes the $B$ systems according to $\pi\in S_k$. This shows that indeed $(i)\Leftrightarrow (ii)$. Theorem~\ref{GClosed} guarantees that $(ii)\Leftrightarrow(iv)$. Also, the implications $(iv)\Rightarrow (iii)\Rightarrow (i)$ follow by definition. This completes the proof.
\end{proof}

Due to this equivalence result, the above definition of general $k$-extendibility does not add much to our discussion. We therefore do not consider it further, except very briefly in Remark~\ref{comparison Bhat rem} below.

\section{Gaussian extendibility: degenerate cases}

In this section, we argue that the proof of Theorem~\ref{G k ext thm} remains valid also in the degenerate cases where $V_A-\I \Omega_A$ is not invertible and thus some of the intermediate statements in the main text cease to hold as written. Namely,~\eqref{k ext necessary 1} and~\eqref{simplified k ext necessary 1} do not seem to make sense when $V_A-\I \Omega_A$ does not possess an inverse. As it turns out, the right way to interpret these conditions is via a regularization procedure. For instance,~\eqref{simplified k ext necessary 1} is said to hold if it holds for all $\epsilon>0$ when one performs the substitution  $V_A\mapsto V_A+\epsilon\mathds{1}_A$.

A regularization procedure can be used to define the Schur complement with respect to a non-invertible block: given a positive semidefinite matrix $M$ as in~\eqref{eq-Schur_complement}, where $P$ is not necessarily invertible, we define $M/P$ as
\begin{equation}
    M/P \coloneqq \lim_{\epsilon\to 0^+} \left( Q - Z^\dag (P + \epsilon \mathds{1})^{-1} Z\right) ,
    \label{generalized Schur}
\end{equation}
provided that such a limit exists. Observe that this happens if and only if $\mathrm{im}(Z)\perp \ker(P)$, where $\mathrm{im}$ denotes the image (or range), and orthogonality is between subspaces. When this is the case, one still has $M/P = Q- Z^\dag P^{-1} Z$ provided that the inverse is taken on the support. This definition of generalized Schur complement fits well into the positive semidefiniteness condition in~\eqref{eq-Schur_complement}. In fact, note that
\begin{equation}\label{eq-Schur_complement epsilon}
	M = \begin{pmatrix} P & Z \\ Z^\dag & Q\end{pmatrix} \geq 0\ \Leftrightarrow\ P \geq 0 ,\ M/(P+\epsilon \mathds{1})\coloneqq Q - Z^\dag (P+\epsilon \mathds{1})^{-1} Z \geq 0\quad \forall\ \epsilon>0\, .
\end{equation}
Hence, one can still claim that the fundamental equivalence in~\eqref{eq-Schur_complement} holds formally, provided that Schur complements are intended as generalized. It is understood that on the r.h.s.\ of \eqref{eq-Schur_complement} one still requires that $M/P$ actually exists. We now review the proof of Theorem~\ref{G k ext thm} in light of these considerations.

\begin{proof}[Proof of Theorem~\ref{G k ext thm} (complete version)]
In what follows, we assume that $k\geq 2$. From Theorem~\ref{GClosed}, we know that a Gaussian state $\rho_{AB}^\G$ is $k$-extendible if and only if it has a Gaussian $k$-extension. If $\rho_{AB}^\G$ has vanishing first moments, then the same is true for any extension, by Lemma~\ref{simple lemma}. Hence, searching for a Gaussian $k$-extension amounts to asking whether there exists a legitimate QCM $\widetilde{V}_{AB_1\ldots B_k}$ of the form as in~\eqref{V AB1...Bk} that satisfies the bona fide condition
\begin{equation}
    \widetilde{V}_{AB_1\ldots B_k} \geq \I \Omega_{AB_1\ldots B_k} = (\I \Omega_A) \oplus (\I \Omega_{B_1}) \oplus \cdots \oplus  (\I \Omega_{B_k})\, .
    \label{Heisenberg V AB1...Bk}
\end{equation}
In fact, it is not difficult to verify that Gaussian states with zero mean possess the same symmetries as their QCMs. Specifically, any Gaussian state with a QCM of the form as in~\eqref{V AB1...Bk} is both invariant under permutation of any two $B$ systems and such that its reduction to $AB_1$ coincides with the original state with QCM $V_{AB}$ as in~\eqref{V AB}.

We now rephrase~\eqref{Heisenberg V AB1...Bk} using the suitably regularized conditions in~\eqref{eq-Schur_complement}, as given explicitly by~\eqref{eq-Schur_complement epsilon}. Remember that $V_A\geq \I \Omega_A$ holds by hypothesis since the reduced state on $A$ is a legitimate density matrix. Making the dependence on the regularizing parameter $\epsilon>0$ explicit for clarity, we thus obtain that
\begin{align}
0 &\leq \Big( \widetilde{V}_{AB_1\ldots B_k} + \epsilon \mathds{1}_A \oplus 0_{B_1\ldots B_k} - \I \Omega_{AB_1\ldots B_k} \Big)\, \Big/\, \Big( V_A + \epsilon \mathds{1}_A - \I \Omega_A \Big) \\
&= \begin{pmatrix} V_B - \I \Omega_B & Y  & \ldots & Y \\[-1ex] Y & V_B - \I \Omega_B & \ddots & \vdots \\[-1ex] \vdots & \ddots & \ddots & Y \\ Y & \ldots & Y & V_B - \I \Omega_B \end{pmatrix} - \begin{pmatrix} X^\t \\ X^\t \\ \vdots \\ X^\t \end{pmatrix} \left( V_A + \epsilon \mathds{1}_A - \I\Omega_A\right)^{-1} \begin{pmatrix} X & X & \ldots & X \end{pmatrix} , \label{Schur big matrix}
\end{align}
to be obeyed for all $\epsilon>0$. We can conveniently write the above inequality by making the identification $\bigoplus_{j=1}^k \mathds{R}^{2n_B} = \mathds{R}^k \otimes \mathds{R}^{2n_B}$ in terms of the underlying vector spaces. Introducing $\ket{+}\coloneqq \frac{1}{\sqrt{k}}\sum_{j=1}^k \ket{j}\in\mathds{R}^k$, the inequality~\eqref{Schur big matrix} becomes
\begin{align}
0 &\leq \mathds{1}_k \otimes (V_B - \I \Omega_B - Y) + k \ketbra{+} \otimes Y - k\ketbra{+}\otimes \left( X^\t \left( V_A + \epsilon \mathds{1}_A - \I \Omega_A\right)^{-1} X\right) \\
&= \left( \mathds{1}_k - \ketbra{+}\right) \otimes (V_B - \I \Omega_B - Y) + \ketbra{+}\otimes \left( V_B + (k-1) Y - \I \Omega_B - k X^\t \left( V_A + \epsilon \mathds{1}_A - \I \Omega_A\right)^{-1} X \right)
\end{align}
Now, since $\ketbra{+}$ and $\mathds{1} - \ketbra{+}$ are projectors onto orthogonal subspaces, the above relation is equivalent to
\begin{align}
    V_B - \I \Omega_B - Y &\geq 0\, , \label{intermediate condition 1}\\
    V_B + (k-1) Y - \I \Omega_B - k X^\t \left( V_A + \epsilon \mathds{1}_A - \I \Omega_A\right)^{-1} X &\geq 0\, . \label{intermediate condition 2}
\end{align}
Introducing the alternative parametrization $\Delta \coloneqq V_B - Y$, we can rephrase this as
\begin{equation}
\I \Omega_B\leq \Delta_B\leq \frac{k}{k-1}\left( V_B - X^\t \left( V_A + \epsilon \mathds{1}_A - \I \Omega_A\right)^{-1} X\right) - \frac{1}{k-1}\, \I \Omega_B\qquad \forall\ \epsilon>0\, ,
\label{condition k ext epsilon}
\end{equation}
reproducing~\eqref{k ext necessary 1}. Using again~\eqref{eq-Schur_complement epsilon} -- this time backwards -- we see that~\eqref{condition k ext epsilon} is equivalent to the existence of a real matrix $\Delta_B\geq \I\Omega_B$ such that~\eqref{k ext necessary 2} is obeyed.

As already mentioned in the main text, to see that~\eqref{k ext necessary 2} implies~\eqref{simplified k ext necessary 2} one substitutes the (complex conjugate) bona fide condition $\Delta_B \geq -\I\Omega_B$ into~\eqref{k ext necessary 2}. Observe that this is possible since $\Delta_B$ is real.

Now, the problem is to prove that~\eqref{simplified k ext necessary 2} is also sufficient to guarantee the existence of a real $\Delta\geq \I \Omega_B$ that satisfies~\eqref{k ext necessary 2} when $n_B=1$. In order to do this, it suffices to prove that~\eqref{simplified k ext necessary 2} is equivalent to~\eqref{condition k ext epsilon}. To this end, we employ~\cite[Lemma~7]{revisited}, which states that given two $2\times 2$ Hermitian matrices $M,N$, there exists a real matrix $R$ such that $M\leq R\leq N$ if and only if both $M\leq N$ and $M^*\leq N$ hold true, with $M^*$ being the complex conjugate of $M$. Since  all matrices in~\eqref{condition k ext epsilon} are $2\times 2$ when $n_B=1$, we can rephrase it as
\begin{equation}
\pm \I \Omega_B \leq \frac{k}{k-1}\left( V_B - X^\t \left( V_A + \epsilon \mathds{1}_A - \I \Omega_A\right)^{-1} X\right) - \frac{1}{k-1}\, \I \Omega_B\qquad \forall\ \epsilon>0\, ,
\end{equation}
which upon elementary algebraic manipulations translates to the following two conditions:
\begin{equation}
\begin{aligned}
V_B - X^\t \left( V_A + \epsilon \mathds{1}_A - \I \Omega_A\right)^{-1} X &\geq \I \Omega_B\, , \\
\frac{k}{k-1}\left( V_B - X^\t \left( V_A + \epsilon \mathds{1}_A - \I \Omega_A\right)^{-1} X \right) &\geq - \frac{k-2}{k-1}\, \I \Omega_B\, ,
\end{aligned}
\qquad\quad \forall\ \epsilon>0\, .
\label{simplified condition k ext epsilon}
\end{equation}
The first inequality follows from the bona fide condition $V_{AB}\geq \I\Omega_{AB}$ via an application of~\eqref{eq-Schur_complement epsilon}, while the second is equivalent to~\eqref{simplified k ext necessary 2} again via~\eqref{eq-Schur_complement epsilon}.

Finally, the fact that~\eqref{k ext necessary 2} and~\eqref{simplified k ext necessary 2} fail to be equivalent already for $n_A=n_B=k=2$ is demonstrated by the example of the Werner--Wolf state of~\cite{Werner01}, as the discussion in the next section shows.
\end{proof}

\begin{remark} \label{comparison Bhat rem}
We take the chance here to draw a thorough comparison between our results and techniques and those of~\cite{Bhat16}. The starting point is the structure of the $k$-extended QCM in~\eqref{V AB1...Bk}, to be compared with~\cite[Eq.~(2.1)]{Bhat16}. We see that the fact that our matrix $Y$ (to be identified with their $\theta_k$) needs to be symmetric was not recognized in~\cite{Bhat16} as descending directly from the required symmetry of the extended state $\rho_{AB_1\ldots B_k}$ under the exchange of any two $B$ systems. We have instead proved this explicitly (Lemma~\ref{simple lemma}). Also, we have shown in Corollary~\ref{GClosed stronger} that a Gaussian state is generally $k$-extendible if and only if it has a Gaussian symmetric $k$-extension; this provides another way to show that if a general (possibly non-symmetric) $k$-extension of the state exists, then the corresponding matrix $Y$ that appears in its QCM as in~\eqref{V AB1...Bk} can always be taken to be symmetric.

Following the argument in~\cite{Bhat16}, one notes that the symmetry of $\theta_k$ is recovered via~\cite[Theorem~2.1]{Parthasarathy2015} as a consequence of the bona fide condition \emph{when complete extendibility is assumed}. At this precise point the reasoning in~\cite{Bhat16} ceases to apply to $k$-extendible states with finite $k$, and holds instead only for completely extendible ones. Therefore, while the algebraic manipulations that lead to our~\eqref{intermediate condition 1}--\eqref{intermediate condition 2} are identical to those in~\cite{Bhat16}, and both are indeed elementary applications of known properties of Schur complements, the conclusions that one is allowed to draw from them here are significantly more powerful than those obtained in~\cite{Bhat16}. For example, we are able to derive simple necessary and sufficient conditions for the $k$-extendibility of Gaussian states, solving an outstanding open problem that was explicitly stated as such in~\cite{Bhat16}.

Last but not least, in~\cite{Bhat16} it does not seem to have been observed that the case $n_B=1$ can be solved analytically, yielding a necessary and sufficient condition that resembles PPT-ness. This special case is of paramount physical importance because of its applicability to the theory of quantum communication over single-sender single-receiver Gaussian channels.
\end{remark}

Before we move on, we want to dwell on the possibility of deriving the separability condition for Gaussian states found in~\cite[Theorem~5]{revisited} directly from Theorem~\ref{G k ext thm}. As argued in the main text, due to the fact that $k$-extendibility is a stronger property than $(k-1)$-extendibility, it suffices take the limit $k\to\infty$ of the condition in~\eqref{k ext necessary 2}. Technically, the evaluation of this limit is made less obvious by the fact that the choice of $\Delta_B$ may depend on $k$. This is \textit{a priori} a problem because the set of QCMs is not compact. However, it is not difficult to verify that any $\Delta_B$ satisfying~\eqref{k ext necessary 2} must automatically satisfy also $\Delta_B\leq V_B$. Since the set of QCMs with this property is compact, the sequence of matrices $\Delta_B$ admits a converging subsequence, and we can take the limit on that subsequence~\cite{Bhat16}. This shows that~\cite[Theorem~5]{revisited} is in fact a corollary of our more general Theorem~\ref{G k ext thm}.

\section{Extendibility of two-mode Gaussian states}

We devote this section to the analytical characterization of $k$-extendibility in the simple case of two-mode Gaussian states, as can be deduced from our general Theorem~\ref{G k ext thm}. This is in complete analogy with the analytical characterization of separability given in~\cite{Simon00}. We start by recalling that the QCM of any two-mode Gaussian state can be brought into the normal `Simon' form
\begin{equation}
V_{AB} = \begin{pmatrix} a & & c_+ & \\ & a & & c_- \\ c_+ & & b & \\ & c_- & & b \end{pmatrix}
\label{Simon form}
\end{equation}
by means of local symplectic operations. Here, in compliance with the notation of~\eqref{V AB}, the first two rows and columns pertain to the first mode, and the remaining two to the second. As shown in~\cite{Simon00}, the bona fide condition $V_{AB}\geq \I\Omega_{AB}$ translates to a set of simple inequalities involving the local symplectic invariants $a,b,c_+,c_-$:
\begin{equation}
    V_{AB}\geq \I \Omega_{AB}\quad \Longleftrightarrow\quad \left\{\begin{array}{l} a,b> 0\, ,\ \ ab>c_\pm^2\, , \\ (ab-c_+^2)(ab-c_-^2)+1\geq a^2+b^2+2c_+c_-\geq 2\, . \end{array}\right.
    \label{Simon condition explicit}
\end{equation}
We can make the local symmmetry more explicit by rephrasing these conditions in terms of local symplectic invariants only. Referring to~\eqref{V AB} for notation, we obtain that
\begin{equation}
    V_{AB}\geq \I \Omega_{AB}\quad \Longleftrightarrow\quad \left\{\begin{array}{l} V_{AB}>0\, ,\\ \det (V_{AB}) + 1 \geq \det V_A + \det V_B + 2 \det X \geq 2\, . \end{array}\right.
    \label{Simon condition symplectic}
\end{equation}
By appropriately modifying the second condition~\footnote{Let us sketch the derivation here for readability. First we use the Schur factorization theorem to write $\det V_{AB}=(\det V_A) (\det (V_B-X^\t V_A^{-1} X))$. Then we apply the special determinant formula $\det (M-N)=\det M - \Tr[M\Omega N^\t \Omega^\t] + \det N$, valid for $2\times 2$ real matrices, to $M=V_B$ and $N=X^\t V_A^{-1} X$. Finally, we employ also the $2\times 2$ matrix inversion formula $(\det V_A) V_A^{-1}=\Omega V_A \Omega^\t$. Formula~\eqref{Simon condition symplectic 2} is obtained by putting all together, remembering that $\Omega$ is skew-symmetric.}, one can also rewrite it as
\begin{equation}
    (\det V_A) (\det V_B) - \Tr\left[V_A\Omega X \Omega V_B \Omega X^\t \Omega \right] + \left( 1 - \det X\right)^2 \geq \det V_A + \det V_B \geq 2\left(1- \det X\right)\, ,
    \label{Simon condition symplectic 2}
\end{equation}
which reproduces~\cite[Eq.~(17)]{Simon00}, once one accounts for the extra factor $1/2$ that comes from the different conventions adopted to define QCMs~\footnote{Curiously, the rightmost inequality in~\eqref{Simon condition symplectic 2} -- equivalently, the rightmost inequality in the second line of~\eqref{Simon condition explicit} -- does not appear in~\cite{Simon00}. It is however necessary and not implied by the others, as the example $a=b=1$, $c_+=1-\epsilon=-c_-$ shows.}.

We now write out the conditions for a two-mode Gaussian state with QCM as in~\eqref{Simon form} to be $k$-extendible. In what follows, we usually assume that the given QCM already obeys the bona fide condition.

\begin{proposition} \label{2-mode ext prop}
    Let $V_{AB}\geq \I \Omega_{AB}$ be a two-mode QCM partitioned as in~\eqref{V AB}. Then it is $k$-extendible on $B$ if and only if
    \begin{equation}
        \det V_{AB} \geq \left( 1-\frac2k\right)^2 (\det V_A -1) + \det V_B - 2\left(1-\frac2k\right) \det X\, ,
        \label{2-mode k-ext symplectic 1}
    \end{equation}
    which can be equivalently formulated as
    \begin{equation}
        (\det V_A)(\det V_B) - \Tr\left[ V_A \Omega X \Omega V_B \Omega X^\t \Omega \right] + \left(1-\frac2k+\det X\right)^2 \geq + \left(1-\frac2k \right)^2 \det V_A + \det V_B\, .
        \label{2-mode k-ext symplectic 2}
    \end{equation}
    If the QCM is parametrised as in~\eqref{Simon form}, then~\eqref{2-mode k-ext symplectic 1} and~\eqref{2-mode k-ext symplectic 2} translate to
    \begin{equation}
    (ab-c_+^2)(ab-c_-^2) \geq \left(1-\frac2k\right)^2 (a^2-1) + b^2 - 2 \left(1-\frac2k\right) c_+c_-\, .
    \label{2-mode k-ext explicit}
    \end{equation}
\end{proposition}

\begin{proof}
Without loss of generality, we can prove~\eqref{2-mode k-ext explicit}. Applying the inequality~\eqref{simplified k ext necessary 2} and permuting second and third rows and columns, we see that $k$-extendibility of $V_{AB}$ on $B$ is equivalent to the matrix inequality
\begin{equation}
    \begin{pmatrix} a & c_+ & -\I & 0 \\ c_+ & b & 0 & (1-2/k)\I \\ \I & 0 & a & c_- \\ 0 & -(1-2/k)\I & c_- & b \end{pmatrix} \eqqcolon \begin{pmatrix} Q & -\I D \\ \I D & P \end{pmatrix} \geq 0\, ,
\end{equation}
where the blocks $Q,P,D$ are all $2\times 2$. Since QCMs are positive definite, it automatically holds that $Q,P>0$. Hence, we can use the block positivity conditions in~\eqref{eq-Schur_complement} (nondegenerate case) to rephrase this as
\begin{equation}
    Q \geq D P^{-1} D
\end{equation}
and in turn as
\begin{equation}
    M \coloneqq P^{1/2} D^{-1} Q D^{-1} P^{1/2} \geq \mathds{1}\, .
    \label{M matrix}
\end{equation}
Here, we assumed that $D$ is invertible, i.e., that $k\neq 2$. We will see \textit{a posteriori} that this assumption is anyway immaterial as far as the final result is concerned. Call $\lambda,\mu>0$ the two eigenvalues of the above positive definite matrix $M$. The conditions $\lambda,\mu\geq 1$ are equivalent to $1+\lambda\mu \geq \lambda+\mu\geq 2$. Note that 
\begin{align}
    \lambda\mu &= \det (M) = \frac{\det Q \det P}{\det D^2} = \frac{(ab-c_+^2)(ab-c_-^2)}{\left(1-2/k\right)^2}\, ,\\
    \lambda+\mu &=\Tr[M]=\Tr\left[PD^{-1}QD^{-1}\right] = a^2+\frac{b^2}{\left( 1-2/k\right)^2} - \frac{2c_+c_-}{1-2/k}\, .
\end{align}
We now argue that the condition $\lambda +\mu\geq 2$ is superfluous. In fact, if $c_+c_-< 0$ it is already implied by the rightmost inequality in the second line of~\eqref{Simon condition explicit}. Otherwise, if $c_+c_-\geq 0$ it is well-known~\cite{Simon00} that $V_{AB}$ corresponds to a separable state, hence it is $k$-extendible for all $k$, and $\lambda+\mu\geq 2$ must be satisfied. We are thus left with the condition $1+\lambda\mu \geq \lambda+\mu$, which is exactly~\eqref{2-mode k-ext explicit}. For the special case $k=2$ this reduces to the simple inequality $(ab-c_+^2)(ab-c_-^2) \geq b^2$, which is elementarily seen to be equivalent to $M\geq \mathds{1}$, where $M$ is given by~\eqref{M matrix}. This shows that the working assumption $k\neq 2$ is in fact unnecessary, as claimed. Finally, by rewriting~\eqref{2-mode k-ext explicit} we obtain directly~\eqref{2-mode k-ext symplectic 1}, and from it also~\eqref{2-mode k-ext symplectic 2}. Observe that~\eqref{2-mode k-ext symplectic 2} can be derived from~\eqref{2-mode k-ext symplectic 1} by elementary linear algebra techniques, as explained in~\cite{Simon00} and reviewed above.
\end{proof}

As a final remark, observe that in the limit $k\to\infty$ either of the inequalities~\eqref{2-mode k-ext symplectic 1}--\eqref{2-mode k-ext explicit} reduces to the separability condition in~\cite[Eq.~(19)]{Simon00}, as expected.

\section{The Werner--Wolf state is not 2-extendible}

This subsection is devoted to the analysis of the Werner--Wolf bound entangled Gaussian state of~\cite{Werner01} from the point of view of $k$-extendibility. This bipartite Gaussian state of a system composed of $2+2$ modes is particularly interesting because it can be shown to be $2$-unextendible yet to obey~\eqref{simplified k ext necessary 2} for all $k$. It thus demonstrates that condition~\eqref{simplified k ext necessary 2} is no longer equivalent to~\eqref{k ext necessary 2} when the local subsystems consist of at least two modes each. We start by recalling that its QCM is given by~\cite[Eq.~(9)]{Werner01}
\begin{equation}
\gamma_{AB} = \begin{pmatrix}
 2 & 0 & 0 & 0 & 1 & 0 & 0 & 0 \\
 0 & 1 & 0 & 0 & 0 & 0 & 0 & -1 \\
 0 & 0 & 2 & 0 & 0 & 0 & -1 & 0 \\
 0 & 0 & 0 & 1 & 0 & -1 & 0 & 0 \\
 1 & 0 & 0 & 0 & 2 & 0 & 0 & 0 \\
 0 & 0 & 0 & -1 & 0 & 4 & 0 & 0 \\
 0 & 0 & -1 & 0 & 0 & 0 & 2 & 0 \\
 0 & -1 & 0 & 0 & 0 & 0 & 0 & 4 \\
\end{pmatrix} .
\label{WW state}
\end{equation}
The above expression in understood to pertain to the following ordering of the four pairs of canonical operators: $x_{A,1},p_{A,1},x_{A,2},p_{A,2},x_{B,1},p_{B,1},x_{B,2},p_{B,2}$. We start with a simple result that connects the simplified condition in~\eqref{simplified k ext necessary 2} with the positive partial transposition (PPT) criterion~\cite{PeresPPT}.

\begin{lemma} \label{lemma PPT}
A QCM $V_{AB}$ obeys~\eqref{simplified k ext necessary 2} if and only if the corresponding Gaussian state is PPT, i.e., if and only if $V_{AB}\geq (\I\Omega_A) \oplus (-\I \Omega_B)$.
\end{lemma}

\begin{proof}
The matrices $V_{AB} - (\I\Omega_A) \oplus \left( - \left(1-\frac2k\right) \I \Omega_B \right)$ are all positive semidefinite by hypothesis. Since positive semidefinite matrices form a closed set, the limit
\begin{equation}
\lim_{k\to\infty} \left\{V_{AB} - (\I\Omega_A) \oplus \left( - \left(1-\frac2k\right) \I \Omega_B \right) \right\} = V_{AB} - (\I \Omega_A)\oplus (-\I \Omega_B)   
\end{equation}
is also positive semidefinite. This is the same as saying that the Gaussian state with QCM $V_{AB}$ is PPT, as shown in~\cite[Eq.~(10)]{Simon00}.
\end{proof}

Thanks to Lemma~\ref{lemma PPT} and leveraging the fact that the Werner--Wolf state is PPT by construction, we know that it obeys~\eqref{simplified k ext necessary 2} for all $k$. However, we now proceed to show that it is not even two-extendible. Let us first establish some technical lemmata.

\begin{lemma} \label{geometric mean classical state lemma}
	Let $\Delta\geq \I\Omega$ be a QCM. Then $\Delta\# (\Omega \Delta \Omega^{\t})\geq \mathds{1}$, where 
	\begin{equation}\label{geometric mean}
		A\#B \coloneqq A^{1/2} \left( A^{-1/2}BA^{-1/2}\right)^{1/2} A^{1/2}
	\end{equation}
	denotes the matrix geometric mean~\cite{geometric-mean, ando79}.
\end{lemma}

\begin{proof}
	Since any QCM $\Delta$ is lower bounded by the QCM of some pure state, and the matrix geometric mean is monotonic in both entries, we can freely assume that $\Delta$ is the QCM of a pure state, i.e., that it is a symplectic matrix. This means that $\Delta\Omega \Delta = \Omega$ (remember that $\Delta=\Delta^{\t}$), which we can alternatively write as $\Omega \Delta \Omega^{\t} = \Delta^{-1}$. Then it is straightforward to see that $\Delta\# (\Omega \Delta \Omega^{\t}) = \Delta\# \Delta ^{-1}=\mathds{1}$, which completes the proof.
\end{proof}

The above result can be interpreted by noticing that the condition $\Delta\# (\Omega \Delta \Omega^{\t})\geq\mathds{1}$ amounts to saying that the Gaussian state with QCM $\Delta\# (\Omega \Delta \Omega^{\t})$ is a convex combination of coherent states, i.e., it is a classical state. Thus, even if $\Delta$ represents a highly squeezed state, taking the above geometric mean ``averages out'' all the squeezing.

\begin{lemma} \label{product inequality lemma}
	For all $2n\times 2n$ QCMs $\Delta\geq \I\Omega$ and all vectors $\ket{v}\in \mathds{C}^{2n}$, we have that
	\begin{equation}\label{product inequality}
		\bra{v}\Delta\ket{v} \bra{v}\Omega\Delta\Omega^{\t}\ket{v} \geq 1.
	\end{equation}
\end{lemma}

\begin{proof}
	The real-valued map $A\mapsto \bra{v}A\ket{v}$ is positive, i.e., it is nonnegative on positive semidefinite matrices. The claim follows from~\cite[Theorem~3]{ando79} and Lemma~\ref{geometric mean classical state lemma}. 
\end{proof}

We are now ready to prove the main result of this section.

\begin{proposition}
    The Werner--Wolf Gaussian state with QCM given by~\eqref{WW state} is not $2$-extendible on the $B$ system.
\end{proposition}

\begin{proof}
Since it can be readily verified that $\gamma_A>\I\Omega_A$ (more precisely, the symplectic spectrum of $\gamma_{A}$ is $\{\sqrt2, \sqrt2\}$), we can rephrase~\eqref{k ext necessary 2} as~\eqref{k ext necessary 1} without the need to consider generalized Schur complements. Computing the r.h.s.\ of~\eqref{k ext necessary 1} for $V_{AB}=\gamma_{AB}$ and for $k=2$ yields
\begin{equation}
H_B \coloneqq 2 \left( \gamma_B - X_\gamma^\t (\gamma_A-\I \Omega_A)^{-1} X_\gamma \right) - \I\Omega_B = \begin{pmatrix} 2 & -\I & 0 & 2 \I \\ \I & 4 & 2 \I & 0 \\ 0 & -2 \I & 2 & -\I \\ -2 \I & 0 & \I & 4 \end{pmatrix} .
\label{RHS for WW state}
\end{equation}
We want to show that there does not exist a matrix $\Delta_B$ such that $\Delta_B\leq H_B$. It is straightforward to see that the normalized vector
\begin{equation}
	\ket{v}\coloneqq \frac{1}{\sqrt{ 2\left(6+\sqrt6\right)}} \begin{pmatrix} -\left(1+\sqrt{6}\right) \I & -1 & 0 & 2 \end{pmatrix}^{\t}
\end{equation}
is an eigenvector of $H_B$ with corresponding eigenvalue $3-\sqrt{6}$. If there existed a real $\Delta_B$ such that $\I \Omega_B\leq \Delta_B\leq H_B$ then by Lemma~\ref{product inequality lemma} we would obtain
\begin{equation}
	1\leq \braket{v|\Delta_B|v} \braket{v|\Omega_B \Delta_B\Omega_B^{\t}|v} \leq \braket{v | H_B|v}\braket{v|\Omega_B H_B \Omega_B^{\t}|v} = \left(3-\sqrt{6}\right) \frac{10}{\left(1+\sqrt{6}\right)^2} < 0.463,
\end{equation}
which is a contradiction. Hence, $\gamma_{AB}$ does not satisfy~\eqref{k ext necessary 2} for any $\Delta_B\geq \I\Omega_B$, implying that the Werner--Wolf state is not $2$-extendible.
\end{proof}

\section{Extendibility of Gaussian channels}

We now provide further details of the $k$-extendibility of single-sender, single-receiver Gaussian channels. By such a $k$-extendible channel, as stated in the main text, we mean that it can be implemented as a broadcast channel from a single sender to $k$ receivers, such that the reduced channel from the sender to any one of the receivers is the same as the original channel. Recall that a Gaussian channel $\mathcal{N}$ with $n$ input and $m$ output modes is uniquely characterized by a pair of real matrices $X,Y$, where $X$ is $2m\times 2n$ and $Y$ is $2m\times 2m$, and a real vector $\delta\in \mathbb{R}^{2m}$, such that $Y+\I \Omega \geq \I X \Omega X^{\intercal}$. Since a Gaussian channel sends Gaussian states to Gaussian states, its action can be described directly at the level of the mean vector  and covariance matrix:
\begin{equation}
\mathcal{N}:\ \left\{ \begin{array}[c]{lcl} V & \longmapsto & XVX^{\intercal}+Y\, , \\ s & \longmapsto & Xs+\delta\, . \end{array} \right.
\end{equation}
In what follows, we  set $\delta=0$ without loss of generality.

Let $A,A'$ be two isomorphic quantum systems, possibly infinite-dimensional. It is well-known that any pure state $\ket{\psi}_{AA'}$ with invertible marginals defines a Choi--Jamio\l kowski isomorphism between the set of quantum channels $\mathcal{N}_{A\to B}$ and the set of bipartite states $\rho_{AB}$ on $AB$ such that $\rho_A=\psi_A\coloneqq \Tr_{A'}\ketbra{\psi}_{AA'}$ (see~\cite{Holevo-CJ} or~\cite{Holevo-CJ-arXiv}). Denoting the Schmidt decomposition of $\ket{\psi}_{AA'}$ by
\begin{equation}
    \ket{\psi}_{AA'} = \sum_{i} \lambda_i^{1/2} \ket{e_i}_A\otimes \ket{f_i}_{A'}\, ,
\end{equation}
the Choi--Jamio\l kowski isomorphism is realized by defining~\cite[Eq.~(6)]{Holevo-CJ-arXiv}
\begin{equation}
    \rho_{AB}^{\mathcal{N}} \coloneqq \left( I_A\otimes \mathcal{N}\right)(\ketbra{\psi}_{AA'})\, . \label{CJ 1}
\end{equation}
Conversely, every state $\rho_{AB}$ such that $\rho_A = \psi_A = \Tr_{A'}\ketbra{\psi}_{AA'}$ identifies a quantum channel $\mathcal{N}_{A\to B}$ via the formulae 
\begin{equation}
    \mathcal{N}(\ketbraa{e_i}{e_j}_A) \coloneqq \frac{1}{\sqrt{\lambda_i \lambda_i}} \Tr_A\left[ \left( \ketbraa{e_j}{e_i}_A\otimes \mathds{1}_B \right) \rho_{AB}\right] \label{CJ 2}
\end{equation}

With this in mind, it is not difficult to realize that a channel $\mathcal{N}_{A\to B}$ is $k$-extendible if and only any (and hence all) of its Choi states $\rho_{AB}^{\mathcal{N}}$ is $k$-extendible on $B$. On the one hand, if $\widetilde{\mathcal{N}}_{A\to B_1\ldots B_k}$ is a $k$-extension of $\mathcal{N}_{A\to B}$, then clearly the corresponding Choi state $\rho_{AB_1\ldots B_k}^{\widetilde{\mathcal{N}}}$ is a $k$-extension of $\rho_{AB}^{\mathcal{N}}$. On the other hand, if $\sigma_{AB_1\ldots B_k}$ is a $k$-extension of $\rho_{AB}^{\mathcal{N}}$, due to the identity $\Tr_{B_1\ldots B_k} [\sigma_{AB_1\ldots B_k}] = \Tr_B \rho_{AB}^{\mathcal{N}} = \psi_A$ we see that $\sigma_{A B_1\ldots B_k}$ is the Choi state of a legitimate quantum channel $\widetilde{\mathcal{N}}_{A\to B_1\ldots B_k}$. By using~\eqref{CJ 2}, it is not difficult to verify that this is indeed a $k$-extension of $\mathcal{N}_{A\to B}$. 

\begin{remark}
If a Gaussian channel is $k$-extendible, then $k$ instances of it can be implemented by the following procedure, related to the approach of~\cite{nogo3,Wolf2007,NFC09}: (1)~prepare a two-mode squeezed vacuum state, (2)~send one share of it through the $k$-extension of the channel, (3)~perform the usual bosonic continuous-variable teleportation protocol~\cite{prl1998braunstein} on the channel input state and the share of the entangled resource not sent into the channel, (4)~broadcast the Bell measurement outcomes to all $k$ receivers, who then correct their local systems. In the limit as the squeezing strength goes to infinity, this simulation of the $k$-extension of the channel converges to the original one with respect to the strong topology.
\end{remark}

The following result is a consequence of the above discussion:

\begin{corollary}
A Gaussian channel $V\mapsto XVX^\intercal + Y$ from $n$ to $m$ modes is $k$-extendible in the sense of~\cite{Pankowski2013, Kaur2018} if and only if there exists a $2m\times 2m$ real matrix $\Delta$ such that
\begin{equation}
\I\Omega\leq \Delta \leq \frac{k}{k-1} \left( Y + \I X \Omega X^\intercal \right) - \frac{1}{k-1}\, \I\Omega\, . \label{k ext Choi state}
\end{equation}
When $m=1$, this is equivalent to
\begin{equation}
Y + \I X\Omega X^\intercal + \left( 1-\frac2k\right) \I \Omega \geq 0\, . \label{k ext Choi state_2}
\end{equation}
If also $n=1=m$, then a simplified equivalent condition that takes into accounts also the complete positivity requirement reads
\begin{equation}
\sqrt{\det Y} \geq 1 - \frac1k + \left| \det X - \frac1k\right| .
\label{eq:single-mode-k-ext-condition}
\end{equation}
\end{corollary}

\begin{proof}
Denote the channel under consideration by $\mathcal{N}$. Then, define its Choi state by $\rho_{AB}^{\mathcal{N}}\coloneqq (\text{id}_A\otimes\mathcal{N}_{A'\to B})(\ket{\psi_r}\bra{\psi_r}^{\otimes n})$, where $\ket{\psi_r}$ is the two-mode squeezed vacuum given by
\begin{equation}
    \ket{\psi_r}=\frac{1}{\cosh(r)}\sum_{j=0}^{\infty}\tanh(r)^j\ket{j,j}.
    \label{two-mode squeezed vacuum}
\end{equation}
The channel $\mathcal{N}$ is $k$-extendible (by definition) if and only if the state $\rho_{AB}^{\mathcal{N}}$ is $k$-extendible, which is true if and only if there exists $\Delta_B$ such that $\Delta_B\geq \I\Omega_B$ and such that the corresponding covariance matrix $V_{AB}^{\mathcal{N}}(r)$ satisfies~\eqref{k ext necessary 1}, where
\begin{equation}
    V_{AB}^{\mathcal{N}}(r)=\begin{pmatrix} \cosh(2r)\mathds{1}_{2n} & \sinh(2r)\Sigma_n X^\intercal \\ \sinh(2r)X\Sigma_n & \cosh(2r)XX^{\intercal}+Y \end{pmatrix},\qquad\Sigma_n=\begin{pmatrix} 1 & 0 \\ 0 & -1 \end{pmatrix}^{\oplus n}.
\end{equation}
Then,
\begin{equation}
    (V_A^{\mathcal{N}}(r)-\I\Omega_A)^{-1}=(\cosh(2r)\mathds{1}_{2n}-\I\Omega_n)^{-1}=\frac{1}{\sinh^2(2r)}(\cosh(2r)\mathds{1}_{2n}+\I\Omega_n).
\end{equation}
Identifying the block $X$ in~\eqref{k ext necessary 1} with the off-diagonal block $\sinh(2r)\Sigma_n X^{\intercal}$ of $V_{AB}^{\mathcal{N}}(r)$, and using $\Sigma_n\Omega_n\Sigma_n=-\Omega_n$, we get
\begin{align}
    &V_B^{\mathcal{N}}(r)-\sinh(2r)X\Sigma_n(\cosh(2r)\mathds{1}_{2n}-\I\Omega_n)^{-1}\Sigma_n X^{\intercal}\sinh(2r)\notag\\
    &=\cosh(2r)XX^{\intercal}+Y-\sinh(2r)X\Sigma_n(\cosh(2r)\mathds{1}_{2n}-\I\Omega_n)^{-1}\Sigma_n X^{\intercal}\sinh(2r)\\
    &=Y+\I X\Omega X^{\intercal},
\end{align}
which leads to~\eqref{k ext Choi state}.

In the case $m=1$, we apply the same calculations above to the condition in~\eqref{simplified k ext necessary 1} in order to obtain~\eqref{k ext Choi state_2}. Finally,~\eqref{eq:single-mode-k-ext-condition} is derived by applying the usual $2\times 2$ determinant formula $M\Omega M^\t=(\det M) \Omega$ to~\eqref{k ext Choi state_2}. The same result can be obtained by resorting directly to Proposition~\ref{2-mode ext prop}.
\end{proof}

We now illustrate the $k$-extendibility conditions of single-mode Gaussian channels. A classification of all such channels into six different categories has been given in~\cite{Holevo2007} (see also~\cite{H12}). Here we can exploit this classification in order to determine $k$-extendibility of all single-mode Gaussian channels because $k$-extendibility of a single-sender, single-receiver channel is invariant under arbitrary input and output unitaries, and the procedure from~\cite{Holevo2007} exploits input and output Gaussian unitaries in order to arrive at the classification.

The categories of single-mode channels from the classification of~\cite{Holevo2007} are labeled as $A_1$, $A_2$, $B_1$, $B_2$, $C$, and $D$. All channels in classes $A_1$, $A_2$, and $D$ are entanglement breaking, as proved in~\cite{Holevo2008}. Thus, these channels are $k$-extendible for all $k\geq 2$. It thus remains to consider the channels in the classes $B_1$, $B_2$, and $C$. Channels in the class $B_1$ have $X = \mathds{1}$ and $Y=(\mathds{1}-\sigma_Z)/2$. Applying the condition in~\eqref{eq:single-mode-k-ext-condition}, we find, for all $k\geq 2$, that channels in this class are not $k$-extendible. This is consistent with the fact that their unconstrained quantum capacity is infinite~\cite{Holevo2007}. The remaining channels are the most important for applications, as stressed in~\cite[Section~3.5]{HG12} and~\cite[Section~12.6.3]{H12}. Channels in the class $B_2$ are called additive-noise channels, and channels in the class $C$ are either thermal channels or  amplifier channels. By applying~\eqref{eq:single-mode-k-ext-condition}, we find the necessary and sufficient conditions for their $k$-extendibility:
\begin{itemize}

\item The \textbf{thermal channel} of transmissivity $\eta\in(0,1)$ and environment thermal photon number $N_S \geq 0$, defined by $X=\sqrt{\eta} \mathds{1}$ and $Y=(1-\eta)(2N_B +1) \mathds{1}$,  is $k$-extendible if  and only if
\begin{equation}
\eta \leq \frac{N_B + 1/k}{N_B+1}     .
\end{equation}
If the channel is a \textbf{pure-loss} channel with $N_B = 0$, then we see that it is $k$-extendible if and only if $\eta \leq 1/k$. 

\item The \textbf{amplifier} channel of gain $G>1$ and environment thermal photon number $N_B \geq 0$, defined by $X=\sqrt{G} \mathds{1}$ and $Y=(G-1)(2N_B+1) \mathds{1}$, is $k$-extendible if and only if
\begin{equation}
 G \geq \frac{N_B + 1 - 1/k}{N_B}   .  
\end{equation}
 If $N_B = 0$, as is the case for the \textbf{pure-amplifier}  channel, then the channel is not $k$-extendible for all $k\geq 2$ and $G > 1$.

\item The \textbf{additive noise} channel defined by $X=\mathds{1}$ and $Y = \xi \mathds{1}$, with noise parameter $\xi > 0$ is $k$-extendible  for $k\geq 2$ if and only if
\begin{equation}
 \xi \geq 2\left( 1-1/k\right).   
\end{equation}
\end{itemize}

As expected, these conditions for $k$-extendibility of the channels imply the entanglement-breaking conditions from~\cite{Holevo2008} in the limit $k\to \infty$. We also recover the conditions for two-extendibility (antidegradability) from~\cite[Eq.~(4.6)]{CGH2006}, for thermal and amplifier channels. See Fig.~\ref{fig:single_mode_regions} for a plot of the parameter space of the single-mode Gaussian channels.

\begin{figure}
    \centering
    \includegraphics{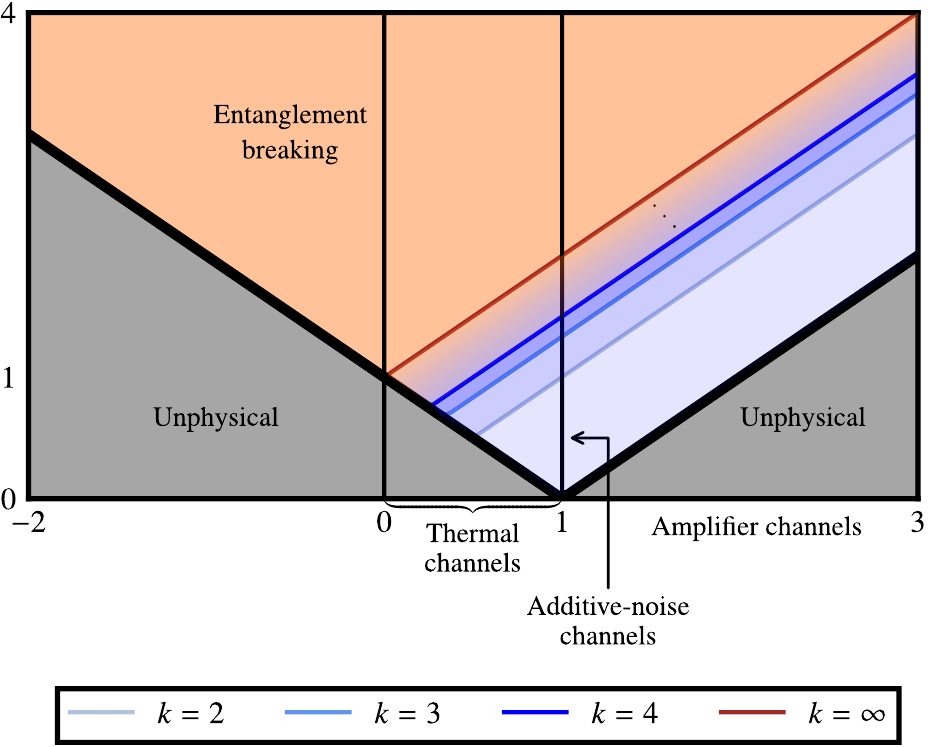}
    \caption{Parameter space of the single-mode Gaussian channels $V\mapsto XVX^{\intercal}+Y$, such that $x\coloneqq\det X$ is on the horizontal axis and $y\coloneqq\sqrt{\det Y}$ is on the vertical axis. All channels occupy the region given by $y\geq |x-1|$, and all entanglement-breaking channels occupy the region given by $y\geq |x|+1$. The $k$-extendibility regions are given from \eqref{eq:single-mode-k-ext-condition} by the inequality $y\geq 1-\frac{1}{k}+\left|x-\frac{1}{k}\right|$.}\label{fig:single_mode_regions}
\end{figure}

The multi-mode \textbf{additive noise} channel defined by $X=\mathds{1}$ and some $Y\geq 0$ is $k$-extendible ($k\geq 2$) if and only if $Y>0$ and $\nu_{\min}(Y)\geq 2\left( 1-1/k\right)$, where $\nu_{\min}$ indicates the minimal symplectic eigenvalue.

\section{Proof of Theorem~\ref{thm-dist_sep_k_ext}}

In this section, we prove Theorem~\ref{thm-dist_sep_k_ext}, which is the statement that any Gaussian state $\rho_{AB}^\G$ of $n=n_A+n_B$ modes that is $k$-extendible satisfies
\begin{align}
	\left\|\,\rho_{AB}^\G-\text{SEP}(A\!:\!B)\,\right\|_1 &\leq\frac{2n}{k} \label{distance k ext_supp} \\
	E_{R,\alpha}\left(\rho_{AB}^\G\right)\leq n \ln\left(1+\frac{\eta_{k,\alpha}}{k-1}\right) &\leq\frac{n\, \eta_{k,\alpha}}{k-1}, \label{E R k ext_supp}
\end{align}
where
\begin{align}
	\left\|\,\rho_{AB}^\G-\text{SEP}(A\!:\!B)\,\right\|_1 &\coloneqq \inf_{\sigma_{AB}\in\text{SEP}(A:B)}\left\|\,\rho_{AB}^\G-\sigma_{AB}\,\right\|_1 \label{distance SEP} \\
	E_{R,\alpha}\left(\rho_{AB}^\G\right) &\coloneqq \inf_{\sigma_{AB}\in\text{SEP}(A:B)}D_\alpha\left(\rho_{AB}^\G\, \big\Vert\,  \sigma_{AB}\right)
\end{align}
are the distances of $\rho_{AB}^\G$ to the set of separable states as measured by the trace norm and the Petz--R\'enyi relative entropy $D_\alpha(\rho\Vert \sigma)\coloneqq \frac{1}{\alpha-1}\ln \Tr[\rho^\alpha \sigma^{1-\alpha}]$, respectively, and
\begin{equation}
\eta_{k,\alpha} = \left\{ \begin{array}{ll} 1 & \text{if $\alpha\leq k+1$,} \\ 2 & \text{otherwise.} \end{array}\right.
\label{eta k alpha}
\end{equation}

The crucial fact that will allow us to derive all these results is the following.

\begin{proposition} \label{separable QCM prop}
	Let $\rho_{AB}^\G$ be a Gaussian state with QCM $V_{AB}$. If it is $k$-extendible on $B$, then the Gaussian state with QCM $\frac{k+1}{k-1}\, V_{AB}$ is separable.
\end{proposition}

\begin{proof}
	Thanks to Theorem~\ref{G k ext thm}, we know that if $\rho_{AB}^{\G}$ is $k$-extendible on $B$, then~\eqref{k ext necessary 2} is satisfied for some $\Delta_B\geq \I \Omega_B$. Now, write the complex conjugate bona fide condition for $V_{AB}$ as $V_{AB}\geq (-\I\Omega_A)\oplus (-\I\Omega_B)$. Adding $\frac{1}{k}$ of this inequality to~\eqref{k ext necessary 2} yields
	\begin{equation}
		\left(1+\frac1k\right) V_{AB} \geq \left(1-\frac1k\right) \left( \I \Omega_A\oplus \Delta_B\right) ,
	\end{equation}
	i.e.,
	\begin{equation}
		\frac{k+1}{k-1}\, V_{AB} \geq  \I \Omega_A\oplus \Delta_B .
	\end{equation}
	By~\cite[Theorem~5]{revisited}, this is equivalent to the separability of the Gaussian state with QCM $\frac{k+1}{k-1}\, V_{AB}$.
\end{proof}

\begin{proof}[Proof of Theorem~\ref{thm-dist_sep_k_ext}] We start by proving~\eqref{distance k ext_supp}. Drawing inspiration from the techniques in~\cite{Koenig2009}, we construct an ansatz for the separable state $\sigma_{AB}$ to be inserted into~\eqref{distance SEP}: namely, let $\sigma_{AB}=\sigma^\G_{AB}$ be the Gaussian state with the same (vanishing) first moments as $\rho_{AB}^\G$ and with QCM $W_{AB}\coloneqq \frac{k+1}{k-1}V_{AB}$. By Proposition~\ref{separable QCM prop}, we know that $\sigma_{AB}^\G$ is separable. The crucial property of $\sigma_{AB}^\G$, however, is that it \emph{commutes} with $\rho_{AB}^\G$. In fact, since the two QCMs are proportional to each other, they can be brought into Williamson form by the \emph{same} symplectic matrix. Comparing~\eqref{V Williamson} with~\eqref{rho canonical}, we deduce that $\rho_{AB}^\G$ and $\sigma_{AB}^\G$ are simultaneously diagonalizable, which is equivalent to them commuting. Naturally, the symplectic eigenvalues of $W_{AB} $ are $\mu_j = \lambda \nu_j$, where we set $\lambda \coloneqq \frac{k+1}{k-1}$ and $\nu_j$ are the symplectic eigenvalues of $V_{AB}$. We can then write:
\begin{align}
	\left\|\,\rho^\G_{AB} - \mathrm{SEP}\,\right\|_1 &\leq \left\|\,\rho^\G_{AB} - \sigma_{AB}^\G\, \right\|_1 \\
	&= \left\| U_S^\dag \left( \bigotimes\nolimits_j \rho^\G(\nu_j) - \bigotimes\nolimits_j \rho^\G(\lambda \nu_j) \right) U_S \right\|_1 \\
	&= \left\| \bigotimes\nolimits_{j=1}^n \rho^\G(\nu_j) - \bigotimes\nolimits_{j=1}^n \rho^\G(\lambda \nu_j)\right\|_1 \\
	&\textleq{1} \sum_{j=1}^n \left\|\, \rho^\G(\nu_j) - \rho^\G(\lambda \nu_j)\,\right\|_1 \\
	&\textleq{2} n \sup_{\nu_j\geq 1} \left\|\, \rho^\G(\nu_j) - \rho^\G(\lambda \nu_j)\,\right\|_1 \\
	&\texteq{3} 2n\, \frac{\lambda -1}{\lambda+1} \\
	&= \frac{2n}{k}\, .
\end{align}
The above derivation can be justified as follows: step~1 descends from a telescopic inequality of the form
\begin{align}
	& \left\|\,\rho_1\otimes\ldots \otimes \rho_N - \sigma_1\otimes\ldots \otimes \sigma_N\, \right\|_1 \notag\\
	&= \left\|\, (\rho_1 - \sigma_1)\otimes \rho_2\otimes\ldots\otimes \rho_N + \sigma_1 \otimes \left( \rho_2\otimes\ldots \otimes \rho_N - \sigma_2\otimes\ldots \otimes \sigma_N \right)\, \right\|_1 \\
	&\leq \|\,\rho_1-\sigma_1\,\|_1 + \left\|\, \rho_2\otimes\ldots \otimes \rho_N - \sigma_2\otimes\ldots \otimes \sigma_N\, \right\|_1 \\[-1.5ex]
	&\, \; \vdots \\[-1.5ex]
	&\leq \sum_j \|\,\rho_j - \sigma_j\,\|_1\, ;
\end{align}
The inequality in step~2 holds true because all symplectic eigenvalues $\nu_j$ must be at least $1$. As for step~3, we read from~\cite[Eq.~(10)]{Koenig2009} that
\begin{equation}
	\left\|\,\rho^\G(\nu) - \rho^\G(\mu)\,\right\|_1 = \max_{l\in \mathds{N}} 2\left\{ \left( \frac{\mu-1}{\mu+1} \right)^{l+1} - \left(\frac{\nu-1}{\nu+1} \right)^{l+1} \right\} ,
\end{equation}
which for $\mu=\lambda \nu$ and upon maximization over $\nu\geq 1$ yields~\cite{Koenig2009}
\begin{equation}
\sup_{\nu\geq 1} \left\|\,\rho^\G(\nu) - \rho^\G(\lambda \nu)\,\right\|_1 = 2\, \frac{\lambda -1}{\lambda+1} = \frac2k\, .
\end{equation}
This concludes the proof of~\eqref{distance k ext_supp}.

The inequality in~\eqref{E R k ext_supp} can be proved with an analogous calculation. The relative entropy is even better behaved in this context, as it already factorizes over multiple copies; hence, there is no need for the above telescopic inequality. 
Namely, for a $k$-extendible Gaussian state $\rho_{AB}^\G$ with QCM $V_{AB}$ we can denote as usual with $\sigma_{AB}^\G$ the Gaussian state with the same first moments and QCM $W_{AB}\coloneqq \frac{k+1}{k-1}V_{AB}$. Since $\sigma_{AB}^\G$ is separable by Proposition~\ref{separable QCM prop}, we write
\begin{align}
	E_{R,\alpha}\left(\rho^\G_{AB}\right) &\leq D_\alpha\left( \rho^\G_{AB}\,\Vert\, \sigma^\G_{AB}\right) \\
	&= \sum_{j=1}^n D_\alpha\left( \rho^\G(\nu_j)\,\Vert\, \rho^\G(\lambda \nu_j)\right) \\
	&\leq n \sup_{\nu\geq 1} D_\alpha\left( \rho^\G(\nu)\,\Vert\, \rho^\G(\lambda \nu)\right) ,
\end{align}
where as above $\lambda = \frac{k+1}{k-1}$, and $\nu_1,\ldots, \nu_n$ are the symplectic eigenvalues of $V_{AB}$. Now, observe that the Petz--R\'enyi relative entropy $D_\alpha(\rho\|\sigma)$ is a non-decreasing function of $\alpha\geq 0$~\cite[\S~4.4]{TOMAMICHEL}. For commuting states $[\rho,\sigma]=0$, we have that $\lim_{\alpha\to\infty} D_\alpha(\rho\Vert\sigma) = \ln \left\|\sigma^{-1/2} \rho\sigma^{-1/2}\right\|_\infty=D_\infty(\rho\Vert\sigma)$, where the latter quantity is the max-relative entropy~\cite{Datta2009}. Thus, we have to prove exactly two statements:
\begin{align}
\sup_{\nu\geq 1} D_{k+1}\left( \rho^\G(\nu)\,\Vert\, \rho^\G(\lambda \nu)\right) &= \ln \left( \frac{k}{k-1} \right) , \label{claim 1} \\
\sup_{\nu\geq 1} D_{\infty}\left( \rho^\G(\nu)\,\Vert\, \rho^\G(\lambda \nu)\right) &= \ln \left( \frac{k+1}{k-1} \right) . \label{claim 2}
\end{align}

In general, the Petz--R\'enyi relative entropies between Gaussian states can be computed thanks to the formulae found in~\cite{LL-Renyi}. In the present case, our task is made much easier by the fact that the states commute, and hence the quantum Petz--R\'enyi relative entropies reduce to their classical counterparts. Simple calculations using the expression~\eqref{1 mode canonical} reveal that
\begin{align}
D_{k+1}\left( \rho^\G(\nu)\, \big\Vert\, \rho^\G(\lambda\nu) \right) &= \frac{1}{k} \left( \ln 2 - \ln \left( (\nu+1)^{k+1} (\lambda\nu+1)^{-k} - (\nu-1)^{k+1} (\lambda\nu-1)^{-k} \right)\right) , \label{k+1 relative entropy 1 mode canonical} \\
D_\infty\left( \rho^\G(\nu)\, \big\Vert\, \rho^\G(\lambda\nu) \right) &= \ln \left(\frac{\lambda\nu+1}{\nu+1}\right) . \label{max relative entropy 1 mode canonical}
\end{align}

We start by proving~\eqref{claim 1}. Let us define the function
\begin{equation}
f_k(\nu)\coloneqq (\nu+1)^{k+1} (\lambda\nu+1)^{-k} - (\nu-1)^{k+1} (\lambda\nu-1)^{-k} ,
\label{f}
\end{equation}
where $\lambda=\frac{k+1}{k-1}$. We will now show that $f_k$ is monotonically increasing in $\nu\geq 1$ for all $k>1$. In fact, computing its derivative, one obtains that 
\begin{align}
f'_k(\nu) &= \lambda (\nu^2-1) \frac{(\nu+1)^{k-1}}{(\lambda\nu+1)^{k+1}}\left\{ 1 - \left(\frac{\nu-1}{\nu+1}\right)^{k-1}\left( \frac{\lambda\nu+1}{\lambda\nu-1}\right)^{k+1} \right\} \label{f' eq1}\\
&= \lambda (\nu^2-1) \frac{(\nu+1)^{k-1}}{(\lambda\nu+1)^{k+1}}\left\{ 1 - e^{-(k-1) g(\lambda, \nu)} \right\} ,
\label{f' eq2}
\end{align}
where
\begin{equation}
g(\lambda, \nu) \coloneqq \ln \left( \frac{\nu+1}{\nu-1} \right) - \lambda \ln \left( \frac{\lambda\nu+1}{\lambda\nu-1} \right) .
\label{g}
\end{equation}
Now, we claim that $g(\lambda,\nu) > 0$ for all $\lambda>1$ and $\nu\geq 1$. In fact, $g(1,\nu)\equiv 0$, and
\begin{align}
\frac{\partial g(\lambda,\nu)}{\partial \lambda} &= \frac{2\lambda \nu}{\lambda^2\nu^2-1} + \ln\left( \frac{\lambda\nu-1}{\lambda\nu+1}\right) \\
&= \frac{2\lambda \nu}{\lambda^2\nu^2-1} + \ln\left( 1- \frac{2}{\lambda\nu+1}\right) \\
&> 0\, .
\end{align}
Here, the last inequality is a consequence of the elementary relation $\ln (1-x) + \frac{x(2-x)}{2(1-x)} > 0$, valid for all $0<x<1$. This can in turn be proved by using for instance a power series expansion:
\begin{align}
\ln (1-x) &= - \sum_{r=1}^\infty \frac{x^r}{r} \\
&> - x -\sum_{r=2}^\infty \frac{x^r}{2} \\
&= - x - \frac{x^2}{2} \sum_{r=0}^{\infty} x^r \\
&= -x - \frac{x^2}{2} \frac{1}{1-x} \\
&= - \frac{x(2-x)}{2(1-x)}\, .
\end{align}
We have thus shown that $g(\lambda,\nu) > 0$ for all $\lambda>1$ and all $\nu\geq 1$. Going back to~\eqref{f' eq2}, this implies that $f'_k(\nu) > 0$ for $\nu>1$ and hence that $f_k$ is monotonically increasing whenever $k>1$. Via~\eqref{k+1 relative entropy 1 mode canonical}, this amounts to saying that $D_{k+1}\left( \rho^\G(\nu)\, \big\Vert\, \rho^\G(\lambda\nu) \right)$ is decreasing in $\nu\geq 1$, in turn entailing that
\begin{equation}
\sup_{\nu\geq 1} D_{k+1}\left( \rho^\G(\nu)\, \big\Vert\, \rho^\G(\lambda\nu) \right) = D_{k+1}\left( \rho^\G(1)\, \big\Vert\, \rho^\G(\lambda) \right) = \ln \left( \frac{\lambda+1}{2}\right) = \ln \left( \frac{k}{k-1}\right) .
\end{equation}
This proves~\eqref{claim 1}.

We now turn to the proof of~\eqref{claim 2}, which fortunately can be obtained much more straightforwardly from~\eqref{max relative entropy 1 mode canonical}. Noting that $\nu\mapsto \frac{\lambda\nu+1}{\nu+1}$ is monotonically increasing for $\nu\geq 1$ because $\lambda>1$, one finds that
\begin{equation}
\sup_{\nu\geq 1} D_{\infty}\left( \rho^\G(\nu)\,\Vert\, \rho^\G(\lambda \nu)\right) = \lim_{\nu\to\infty} \ln \left( \frac{\lambda\nu+1}{\nu+1} \right) = \ln \lambda = \ln \left( \frac{k+1}{k-1}\right) ,
\end{equation}
proving~\eqref{claim 2}.
\end{proof}

\section{Optimality of the bounds in Theorem~\ref{thm-dist_sep_k_ext}}

The purpose of this section is to prove that the bounds we established in Theorem~\ref{thm-dist_sep_k_ext} are in some sense optimal, at least in some regimes. We  show the following.

\begin{theorem} \label{thm optimality}
Let $k\geq 2$ be fixed. Then there exists a two-mode $k$-extendible Gaussian state $\rho^{\!\G}$ such that
\begin{equation}\label{tightness trace distance}
	\left\|\,\rho^{\!\G} - \mathrm{SEP}\,\right\|_1 \geq 2 \left\|\,\rho^{\!\G} - \mathrm{SEP}\,\right\|_\infty \geq \frac{1}{2k-1}.
\end{equation}
Moreover, for all positive integers $m$ there is a family of $(m+m)$-mode bipartite $k$-extendible Gaussian states $\rho^\G_{k,m}(r)$ such that
\begin{equation}\label{tightness relative entropy}
	E_{R,1}\left(\rho^\G_{k,m}(r)\right) \geq E_D\left(\rho^\G_{k,m}(r)\right) \geq m \ln \frac{k}{k-1} - o(1)
\end{equation}
as $r\to\infty$, where $E_D$ denotes the distillable entanglement.
\end{theorem}

\begin{remark}
The above Theorem~\ref{thm optimality} shows that (a) the bound in~\eqref{distance k ext} is tight for $n=2$ and for all $k$ up to a universal multiplicative constant of $1/8$; (b) the bound in~\eqref{E R k ext} is tight for all balanced systems ($n_A=n_B$), for all $k$, and for $\alpha\geq 1$, up to a universal multiplicative constant of either $1/2$ (if $\alpha\leq k+1$) or $1/4$ (if $\alpha>k+1$).
\end{remark}

\begin{proof}[Proof of Theorem~\ref{thm optimality}]
We start by proving~\eqref{tightness trace distance}. Let $\ket{\psi_r}=\frac{1}{\cosh(r)}\sum_{j=0}^{\infty}\tanh(r)^j\ket{j,j}$ be the two-mode squeezed vacuum state. Consider a passive symplectic unitary $U=U_{B_1\ldots B_k}$ acting on $k$ modes $B_1,\ldots, B_k$ so that $U^\dag b_1U = \frac{b_1+\ldots+b_k}{\sqrt{k}}$, where $b_j$ is the creation operator associated with the $j$-th mode. Define
\begin{equation}\label{rho r k}
	\rho_{k}^\G(r) = \left(\rho_{k}^\G(r)\right)_{AB_1} \coloneqq \Tr_{B_2\ldots B_k} \left[ U_{B_1\ldots B_k} \left(\ketbra{\psi_r}_{AB_1} \otimes \bigotimes\nolimits_{j=2}^{k} \ketbra{0}_{B_j}\right) U_{B_1\ldots B_k}^\dag \right]
\end{equation}
Clearly, $\rho_{k}^\G(r)$ is $k$-extendible by construction. In fact, it is easy to verify that the state inside the partial trace at the right-hand side of~\eqref{rho r k} is a symmetric extension of it. Another elementary property of the above state is that the passive unitary used to define it acts as an effective attenuator of parameter $\frac{1}{k}$, according with the definition in~\eqref{attenuator}. This means that
\begin{equation}\label{rho r k attenuator}
	\rho_{k}^\G(r) = \left(I\otimes \mathcal{L}_{1/k}\right)\left(\psi_r\right) .
\end{equation}

To estimate $\left\|\,\rho_{k}^\G(r)-\mathrm{SEP}\,\right\|_1$, we first observe that since for all traceless operators $X$ it holds that $\|X\|_1\geq 2\|X\|_\infty$, one can give the lower bound 
\begin{equation}
    \left\|\,\rho_{k}^\G(r)-\mathrm{SEP}\,\right\|_1\geq 2 \left\|\,\rho_{k}^\G(r)-\mathrm{SEP}\,\right\|_\infty\, .
\end{equation}
We now remember that for all bipartite pure states $\ket{\Psi}$ with maximal Schmidt coefficient $\lambda_1(\Psi)$ one has that $\bra{\Psi}\sigma\ket{\Psi}\leq \lambda_1(\Psi)$ for all separable states $\sigma$. Then,
\begin{align}
	\left\|\,\rho_{k}^\G(r)-\mathrm{SEP}\,\right\|_\infty &= \inf_{\sigma\in\mathrm{SEP}} \left\|\,\rho_{k}^\G(r)-\sigma\,\right\|_\infty \\
	&= \inf_{\sigma\in\mathrm{SEP}} \sup_{\ket{\Psi}} \left| \braket{\Psi|\left(\rho_{k}^\G(r)-\sigma\right)| \Psi}\right| \\
	&\geq \sup_{\ket{\Psi}} \left(\braket{\Psi|\,\rho_{k}^\G(r)\,|\Psi} - \lambda_1(\Psi) \right)
\end{align}
Choosing $\ket{\Psi}$ in the family of two-mode squeezed vacua, i.e., $\ket{\Psi}=\ket{\psi_s}$, one obtains
\begin{equation}\label{estimate distance SEP}
	\left\|\,\rho_{k}^\G(r)-\mathrm{SEP}\,\right\|_\infty \geq \sup_s \left( \braket{\psi_s|\, \rho_{k}^\G(r)\, |\psi_s} - \frac{1}{\cosh(s)^2} \right)
\end{equation}

To proceed further, we need to evaluate the matrix element $\bra{\psi_s}\rho_{k}^\G(r)\ket{\psi_s}$. This can be easily done by means of~\eqref{rho r k attenuator} and the Kraus representation in~\eqref{attenuator Kraus}, which together yield
\begin{align}
	\braket{\psi_s|\, \rho_{k}^\G(r)\,|\psi_s} &= \braket{\psi_s| \left( I\otimes \mathcal{L}_{1/k}\right)\left(\psi_r\right) |\psi_s} \\
	&= \frac{1}{\cosh(r)^2 \cosh(s)^2}\sum_{j,\ell=0}^\infty (\tanh(r)\tanh(s))^{j+\ell} \braket{j|\,\mathcal{L}_{1/k}(\ketbraa{j}{\ell})\,|\ell} \\
	&= \frac{1}{\cosh(r)^2 \cosh(s)^2}\sum_{j,\ell=0}^\infty (\tanh(r)\tanh(s))^{j+\ell} \left(\frac{1}{k}\right)^{(j+\ell)/2} \\
	&= \frac{1}{\cosh(r)^2 \cosh(s)^2} \left(\sum_{j=0}^\infty (\tanh(r)\tanh(s))^j \left(\frac{1}{k}\right)^{j/2}\right)^2 \\
	&= \frac{1}{\cosh(r)^2 \cosh(s)^2} \frac{1}{\left(1- \tanh(r)\tanh(s) k^{-1/2}\right)^2} \\
	&= \frac{k}{\left(\sqrt{k} \cosh(r) \cosh(s)- \sinh(r)\sinh(s)\right)^2}
\end{align}
Using the above expression one can verify that for all fixed $s$
\begin{equation}\label{optimisation r}
	\sup_r \bra{\psi_s} \rho_{k}^\G(r) \ket{\psi_s} = \frac{k}{k \cosh(s)^2 - \sinh(s)^2}
\end{equation}
Putting all together, one obtains
\begin{align}
	\sup_r \left\|\rho_{k}^\G(r)-\mathrm{SEP}\right\|_\infty &\textgeq{1} \sup_{r,s} \left( \bra{\psi_s} \rho_{k}^\G(r) \ket{\psi_s} - \frac{1}{\cosh(s)^2} \right) \\
	&\texteq{2} \sup_s \left( \sup_r \bra{\psi_s} \rho_{k}^\G(r) \ket{\psi_s} - \frac{1}{\cosh(s)^2} \right) \\
	&\texteq{3} \sup_s \left( \frac{k}{k \cosh(s)^2 - \sinh(s)^2} - \frac{1}{\cosh(s)^2} \right) \\
	&= \sup_s \frac{\tanh(s)^2}{k \cosh(s)^2 - \sinh(s)^2} \\
	&\texteq{4} \frac{1}{\left( \sqrt{k} + \sqrt{k-1}\right)^2} \\
	&\textgeq{5} \frac{1}{4k-2}.
\end{align}
The above derivation can be justified as follows. Step~1 is obtained from~\eqref{estimate distance SEP} by taking the supremum over $r$. In step~2 we exchanged the order of the suprema over $r$ and $s$. Step~3 comes from~\eqref{optimisation r}. For step~4, we introduce the parameter $x=\sinh(s)^2$, so that 
\begin{equation}
	\sup_s \frac{\tanh(s)^2}{k \cosh(s)^2 - \sinh(s)^2} = \sup_{x\geq 0} \frac{x}{(1+x)((k-1)x+k)} = \frac{1}{\left( \sqrt{k} + \sqrt{k-1}\right)^2} ,
\end{equation}
where the last equality is a consequence of the fact that $f_k(x)\coloneqq \frac{x}{(1+x)((k-1)x+k)}$ achieves its maximum for $x=\sqrt{\frac{k}{k-1}}$. Finally, step~5 rests on the fact that $4k-2-\left( \sqrt{k} + \sqrt{k-1}\right)^2 = \left( \sqrt{k} - \sqrt{k-1}\right)^2\geq 0$. Observe that for $k>1$ the combined supremum in $r,s$ is achieved at $(r,s)=(r_0,s_0)$, where $r_0$ and $s_0$ are the unique positive solutions of the equations
\begin{align}
    \tanh(r_0)^2 &= \frac{1}{\sqrt{k}\left( \sqrt{k} + \sqrt{k-1} \right)} , \\
    \tanh(s_0)^2 &= \frac{\sqrt{k}}{\sqrt{k} + \sqrt{k-1}}\, .
\end{align}
Hence, the existence of a state $\rho^\G$ in the family $\rho_k^\G(r)$ with the property in~\eqref{tightness trace distance} follows.

To establish~\eqref{tightness relative entropy}, for a fixed $m$ we set $\rho_{k,m}^\G(r)\coloneqq \left( \rho_k^\G(r)\right)^{\otimes m}$, which is clearly an $(m+m)$-mode bipartite state on the system $AB$, where $A=A_1\ldots A_m$ represents the collection of all the $m$ subsystems each of which corresponds to the first mode in~\eqref{rho r k}. The distillable entanglement of $\rho_{k,m}^\G(r)$ (which is a well-known lower bound on the relative entropy of entanglement) can be estimated from below with the coherent information $I_{\mathrm{coh}}(B\,\rangle A)_\rho\coloneqq S(\rho_A) - S(\rho_{AB})$ due to the hashing bound in~\cite{devetak2005}:
\begin{equation}
\begin{aligned}
	E_{R,1} \left( \rho_{k,m}^\G(r)\right) &\geq E_D \left( \rho_{k,m}^\G(r)\right) \\
	&= E_D\left( \left( \rho_k^\G(r)\right)^{\otimes m} \right) \\
	&\geq I_{\mathrm{coh}}(B\,\rangle A)_{\left( \rho_k^\G(r)\right)^{\otimes m}} \\
	&= m\, I_{\mathrm{coh}}(B_1\rangle A_1)_{\rho_k^\G(r)} \, .
\end{aligned}
\label{E D lower bounds E R}
\end{equation}
Now, to compute the coherent information $I_{\mathrm{coh}}(B_1\rangle A_1)$ of the state $\rho_k^\G(r)_{A_1B_1}$, we note that its reduced state on $A_1$ coincides with that of the two-mode squeezed vacuum $\ket{\psi_r}_{A_1B_1}$. Therefore, it is simply a one-mode Gaussian state with QCM $\cosh(2r)\mathds{1}_2$, whose entropy can be evaluated using, e.g., the $\alpha=1$ case of the following formula for the R\'enyi-$\alpha$ entropy $S_\alpha(V)$ of a Gaussian state with QCM $V$~\cite[Eq.~(108)]{adesso14}:
\begin{equation}\label{S alpha}
	S_\alpha (V) \coloneqq \left\{ \begin{array}{ll} -\frac{1}{\alpha-1}\, \sum_{j=1}^n \ln \frac{2^\alpha}{\vphantom{\widetilde{E}}\left(\nu_j+1\right)^\alpha-\left(\nu_j-1\right)^\alpha}, & \quad \text{if $\alpha>1$,} \\[3ex] \sum_{j=1}^n\left(\frac{\nu_j+1}{2}\,\ln\frac{\nu_j+1}{2} - \frac{\nu_j -1}{2}\,\ln\frac{\nu_j -1}{2}\right), & \quad \text{if $\alpha=1$.} \end{array}\right.
\end{equation}
Here, $\nu_1,\ldots, \nu_n$ are the symplectic eigenvalues of $V$. We obtain that
\begin{equation}
	S\left(\rho_k^\G(r)_{A_1}\right) = S_1(\cosh(2r)\mathds{1}_2) = \ln \frac{e \cosh(2r)}{2} + o(1)
\end{equation}
in the limit $r\to\infty$. Since the bipartite QCM of $\rho_k^\G(r)_{A_1B_1}$ has only one symplectic eigenvalue different from $1$, and this is equal to $\frac{1+(k-1) \cosh(2 r)}{k}$, we also have that 
\begin{equation}
	S\left(\rho_k^\G(r)_{A_1B_1}\right) = \ln \left(\frac{e}{2}\, \frac{1+(k-1) \cosh(2r)}{k}\right) + o(1)
\end{equation}
as $r\to\infty$. Putting all together, we see that
\begin{equation}
	I_{\mathrm{coh}}(B_1\rangle A_1)_{\rho_k^\G(r)} = \ln \frac{k \cosh(2r)}{1+(k-1) \cosh(2r)} + o(1) = \ln \frac{k}{k-1} + o(1)\, .
\label{lower bound I coh}
\end{equation}
Combining~\eqref{E D lower bounds E R} and~\eqref{lower bound I coh} yields the lower bound in~\eqref{tightness relative entropy}, concluding the proof.
\end{proof}

\section{Proof of Theorem~\ref{thm-Gaussian_EoF}}

This section is devoted to the proof of Theorem~\ref{thm-Gaussian_EoF}. We start by reminding the reader that the R\'{e}nyi-2 Gaussian entanglement of formation, defined by the $\alpha=2$ case of~\eqref{GEoF}, is given by
\begin{equation}\label{GEoF R2}
	E^{\text{G}}_{F,2} \left( \rho_{AB}\right) = \min \left\{ M(\gamma_A):\ \text{$\gamma_{AB}$ pure QCM and $\gamma_{AB}\leq V_{AB}$} \right\},
\end{equation}
where
\begin{equation}\label{M}
	M(V) \coloneqq S_2(V) = \sum_j \ln \nu_j = \frac12 \ln \det V\, .
\end{equation}
In what follows, we will consider the universal function $\varphi:\mathds{R}_+\to \mathds{R}$ given by
\begin{equation}\label{varphi}
	\varphi(x)\coloneqq \frac{e^x+1}{2} \ln \left(\frac{e^x+1}{2} \right)-\frac{e^x-1}{2} \ln \left(\frac{e^x-1}{2} \right).
\end{equation}
Before delving into the proof of Theorem~\ref{thm-Gaussian_EoF}, let us establish a technical lemma that connects the R\'enyi-$2$ Gaussian entanglement of formation with its von Neumann version.

\begin{lemma} \label{GEoF vs R2 GEoF lemma}
	For all bipartite Gaussian states $\rho_{AB}^\G$ on $n_A+n_B$ modes, the entanglement of formation measured in natural units satisfies
	\begin{equation}\label{GEoF vs R2 GEoF}
		E_{F}(\rho_{AB}^{\G}) \leq E_{F}^\G(\rho_{AB}^{\G}) \leq n_A\, \varphi\left(\frac{E_{F,2}^\G(\rho_{AB}^{\G})}{n_A}\right).
	\end{equation}
\end{lemma}

\begin{proof}
	Let $\delta\coloneqq E_{F,2}^\G \left( \rho^\G_{AB} \right)$. By~\eqref{GEoF R2}, there exists a pure QCM $\gamma_{AB} \leq V_{AB}$ such that $M(V_A)=\sum_{j=1}^{n_A} \ln \nu_j = \delta$. Using the readily verified concavity of $\varphi$, one observes that
	\begin{align}
		S_1(V) &= \sum_{j=1}^{n_A} \left(\frac{\nu_j+1}{2}\,\ln\frac{\nu_j+1}{2} - \frac{\nu_j -1}{2}\,\ln\frac{\nu_j -1}{2}\right) \\
		&= \sum_{j=1}^{n_A} \varphi(\ln \nu_j) \\
		&= n_A \sum_{j=1}^{n_A} \frac{1}{n_A}\, \varphi(\ln \nu_j) \\
		&\leq n_A\, \varphi\left( \sum_{j=1}^{n_A} \frac{1}{n_A} \ln \nu_j\right) \\
		&= n_A\, \varphi\left( \frac{\delta}{n_A}\right) .
	\end{align}
	Recalling that
	\begin{align}
		E_{F,\alpha}(\rho_{AB})&=\inf \left\{\sum_i p_i \, S_{\alpha}\left(\psi_{A}^{(i)}\right) :\ \rho_{AB} = \sum_i p_i \psi_{AB}^{(i)} \right\} ,\\
		E^{\text{G}}_{F,\alpha} \left( \rho_{AB}\right)&= \inf \left\{ S_\alpha (\gamma_A):\ \text{$\gamma_{AB}$ pure QCM and $\gamma_{AB}\leq V_{AB}$} \right\} ,
	\end{align}
	for $\alpha=1$ we immediately obtain~\eqref{GEoF vs R2 GEoF}.
\end{proof}

\begin{proof}[Proof of Theorem~\ref{thm-Gaussian_EoF}] The Gaussian R\'enyi-2 entanglement of formation is known to be monogamous on Gaussian states~\cite[Corollary~7]{Lami16}. Calling $\rho^\G_{AB_1\ldots B_k}$ the $k$-extension of $\rho^\G_{AB}$, we then have
\begin{equation}
E_{F,2}^\G \left(\rho^\G_{AB_1\ldots B_k}\right) \geq \sum_{j=1}^k E_{F,2}^\G \left(\rho^\G_{AB_j}\right) = k E_{F,2}^\G \left(\rho^\G_{AB}\right) ,
\end{equation}
i.e.,
\begin{equation}
	E_{F,2}^\G \left(\rho^\G_{AB}\right) \leq \frac1k\, E_{F,2}^\G \left(\rho^\G_{AB_1\ldots B_k}\right) \leq \frac1k\, M(V_A)\, ,
\end{equation}
where the last inequality expresses the fact that $E_{F,\alpha}(\sigma_{AB})\leq S_\alpha (\sigma_A)$ by concavity of the R\'enyi-$\alpha$ entropy; at the level of QCMs, this can also be thought of as a consequence of the fact that $M(\cdot)$ is a monotone function, and any $\gamma_{AB}$ in the set on the right-hand side of~\eqref{GEoF R2} satisfies $\gamma_{A}\leq V_{A}$ and hence $M(\gamma_A)\leq M(V_A)$. Using~\eqref{GEoF vs R2 GEoF} and the fact that $\varphi$ is monotonically increasing we then obtain
\begin{equation}
E_{F}\left(\rho^\G_{AB}\right) \leq E_{F}^\G\left(\rho^\G_{AB}\right) \leq n_A\, \varphi\left( \frac{E_{F,2}^\G \left(\rho^\G_{AB}\right)}{k} \right)  \leq n_A\, \varphi\left( \frac{M(V_A)}{n_A k} \right) ,
\end{equation}
completing the proof.
\end{proof}

\end{document}